\documentclass[runningheads]{llncs}
\usepackage{graphicx}
\usepackage{amsmath,stmaryrd,textcomp,amssymb,enumerate}
\usepackage{mathtools}
\usepackage{thm-restate}
\usepackage{ebproof}
\usepackage{proof}
\usepackage{mathbbol}
\usepackage{amsfonts}
\usepackage{mathrsfs}
\usepackage[inline]{enumitem}
\usepackage{hyperref}
\usepackage{tikz-cd}
\hypersetup{
    colorlinks,
    linkcolor={black},
    urlcolor={black}
}

\makeatletter
\def\labelandtag#1#2{\begingroup
   \def\@currentlabel{#2}%
   \phantomsection\label{#1}\endgroup
}
\makeatother

\renewcommand{\emptyset}{\varnothing}
\renewcommand{\phi}{\varphi}
\newcommand{\at}[1]{\mathrm{#1}}

\newcommand{\myhypertarget}[2]{%
  \phantomsection
  \hypertarget{#1}{#2}%
  \expandafter\gdef\csname targettext@#1\endcsname{#2}%
}
\newcommand{\myhyperlink}[1]{%
  \hyperlink{#1}{\csname targettext@#1\endcsname}%
}

\newcommand{\seq}{\triangleright}

\newcommand{\calculus}[1]{\ensuremath{\mathsf{#1}}}
\newcommand{\subform}[1]{\mathsf{SubF}(#1)}

\newcommand*{\rn}[1]  {\ensuremath{\mathsf{#1}}}

\newcommand*{\weak}{\rn{w}}

\newcommand*{\cont}{\rn{c}}

\newcommand*{\irn}[2][]  {#2_\rn{{I#1}}}
\newcommand*{\ern}[2][]  {#2_\rn{{E#1}}}

\newcommand{\base}[1]{\mathscr{#1}}
\newcommand{\baseB}{\base{B}}
\newcommand{\baseC}{\base{C}}
\newcommand{\baseD}{\base{D}}
\newcommand{\baseE}{\base{E}}
\newcommand{\baseX}{\base{X}}
\newcommand{\baseY}{\base{Y}}

\newcommand{\baseMill}{\base{M}}
\newcommand{\emptybase}{\varnothing}

\newcommand{\baseGeq}{\supseteq}
\newcommand{\baseLeq}{\subseteq}
\newcommand{\suppTor}[1]{\Vdash_{\!\!#1}}
\newcommand{\suppNew}[1]{\Vdash_{\!\!#1}^{\hspace{-1.5mm}\ast}}

\newcommand{\proves}[1]{\vdash_{\!\!#1}}
\newcommand{\deriveBaseM}[1]{\vdash_{\!\!#1}}

\newcommand{\deriveBaseIPL}[1]{\vdash_{\!\!#1}}
\newcommand{\suppM}[2]{\Vdash_{ \!\!#1 }^{ \!\! #2 }}

\newcommand{\flatmill}[1]{{#1}^{\flat}}
\newcommand{\deflatmill}[1]{{#1}^{\natural
}}

\newcommand{\MILL}{\ensuremath{\rm IMLL}}

\newcommand{\At}{\mathbb{A}}

\newcommand{\provesMILL}{\vdash}

\newcommand{\emptymultiset}{\varnothing}
\newcommand{\MILLformula}{\mathsf{Form}_{\MILL}}
\newcommand{\mand}{\otimes}
\newcommand{\mtop}{\mathrm{I}}
\newcommand{\mto}{\mathrel{\multimap}}
\newcommand{\makeMultiset}[1]{[#1]}
\DeclareMathSymbol{\fatcomma}{\mathrel}{bbold}{\lq\,}
\newcommand{\mcomma}{\!\fatcomma\!}
\newcommand{\mmcomma}{\fatcomma}
\newcommand{\calculusMILL}{\ensuremath{\calculus{NIMLL}}}
\newcommand{\degForm}[1]{\mathit{deg}(#1)}

\DeclareFontFamily{U}{min}{}

\DeclareFontShape{U}{min}{m}{n}{<-> udmj30}{}

\AtEndEnvironment{example}{\hfill {\tiny $\blacksquare$}}
\AtEndEnvironment{proof}{\hfill {$\square$}}

\begin{document}

\title{Proof-theoretic Semantics for Intuitionistic Multiplicative 
Linear Logic} 
%
%
\author{Alexander V. Gheorghiu\inst{1}\orcidID{0000-0002-7144-6910} 
\and Tao Gu\inst{1}\orcidID{0000-0001-5749-0758} \and 
David J. Pym\inst{1,2}\orcidID{0000-0002-6504-5838}}
\authorrunning{A. V. Gheorghiu \and T. Gu \and D. J. Pym}
%
\institute{University College London, London WC1E 6BT, UK \and
Institute of Philosophy, University of London,  London WC1H 0AR, UK
\email{\{alexander.gheorghiu.19,tao.gu.18,d.pym\}@ucl.ac.uk}}
\maketitle              

\begin{abstract} 
This work is the first exploration of proof-theoretic semantics for a substructural logic. It focuses on the base-extension semantics (B-eS) for intuitionistic multiplicative linear logic (\MILL{}). 
The starting point is a review of Sandqvist's B-eS for intuitionistic propositional logic (IPL), for which we propose an alternative treatment of conjunction that takes the form of the \emph{generalized} elimination rule for the connective. The resulting semantics is shown to be sound and complete. 
This motivates our main contribution, a B-eS for \MILL{}, in which the definitions of the logical constants all take the form of their elimination rule 
and for which soundness and completeness are established.  

\end{abstract}

\section{Introduction} \label{sec:introduction}


In model-theoretic semantics (M-tS), logical consequence is defined in terms of models; that is, abstract mathematical structures in which propositions are interpreted and their truth is judged. As Schroeder-Heister~\cite{Schroeder2007modelvsproof} explains, in the standard reading given by Tarski~\cite{Tarski1936,Tarski2002}, a propositional formula $\phi$ follows model-theoretically from a context
$\varGamma$ iff every model of $\varGamma$ is a model of $\phi$; that is,
\[
  \begin{array}{r@{\quad}c@{\quad}l}
    \mbox{$\varGamma \models \phi$} & \mbox{iff} & \mbox{for all models $\mathcal{M}$, if  $\mathcal{M} \models \psi$ for all $\psi \in \varGamma$, then $\mathcal{M} \models \phi$}
  \end{array}   
\]
Therefore, consequence is understood as the transmission of truth. Importantly, on this plan, \emph{meaning} and \emph{validity} are characterized is terms of \emph{truth}. 

Proof-theoretic semantics (P-tS) is an alternative approach to meaning and validity in which they are characterized in terms of \emph{proofs} --- understood as objects denoting collections of acceptable inferences from accepted premisses. This is subtle. It is not that one desires a proof system that precisely characterizes the consequences of the logic of interest, but rather that one desires to express the \emph{meaning} of the logical constants in terms of proofs and provability. Indeed, as Schroeder-Heister \cite{Schroeder2007modelvsproof} observes, since no formal system is fixed (only notions of inference) the relationship between semantics and provability remains the same as it has always been --- in particular, soundness and completeness are desirable features of formal systems. Essentially, what differs is that \emph{proofs} serve the role of \emph{truth} in model-theoretic semantics. The semantic paradigm supporting P-tS is \emph{inferentialism} --- the view that meaning 
(or validity) arises from rules of inference (see Brandom~\cite{Brandom2000}).

To illustrate the paradigmatic shift from M-tS to P-tS, consider 
the proposition `Tammy is a vixen'. What does it mean? Intuitively, it means, somehow, `Tammy is female' \emph{and} `Tammy is a fox'. On inferentialism, its meaning is given by the rules,
\[
{\small 
\begin{array}{c}
\infer{\text{Tammy is a vixen}}{\text{Tammy is a fox} & \text{Tammy is female}} \qquad
\infer{\text{Tammy is female}}{\text{Tammy is a vixen}} \qquad \infer{\text{Tammy is a fox}}{\text{Tammy is a vixen}} \\
\end{array}
}
\]
These merit comparison with the laws governing $\land$ in IPL, which justify the sense in which the above proposition is a conjunction:
\[
{\small 
\infer{\phi \land \psi}{\phi & \psi} \qquad
\infer{\phi}{\phi \land \psi} \qquad \infer{\psi}{\phi \land \psi} 
}
\]

There are two major branches of P-tS: 
proof-theoretic validity (P-tV) in the Dummett-Prawitz tradition (see, 
for example,  Schroeder-Heister~\cite{Schroeder2006validity}) and 
base-extension semantics (B-eS) in the sense of, for example,  Sandqvist~\cite{Sandqvist2015hypothesis,Sandqvist2009CL,Sandqvist2015base}. 
The former is a semantics of arguments, and the latter is a semantics 
of a logic, but both are \emph{proof-theoretic semantics}. This paper is concerned with the latter as explained below.

Tennant 
\cite{Tennant78entailment} provides a general motivation for 
P-tV: reading a \emph{consequence} judgement $\varGamma \proves{} \phi$ 
proof-theoretically --- that is, that $\phi$ follows by some 
reasoning from $\varGamma$ --- demands a notion of \emph{valid argument} 
that encapsulates what the forms of valid reasoning are. That is, 
we require explicating the semantic conditions required for an 
argument that witnesses 
\[ 
    \mbox{$\psi_1 , \ldots , \psi_n$; therefore, $\phi$} 
\] 
to be valid. A particular motivation comes from the following programmatic remarks by Gentzen \cite{Gentzen}:
\begin{quote}
The introductions represent, as it were, the ‘definitions’ of the symbols concerned, and the eliminations are no more, in the final analysis, than the consequences of these definitions. This fact may be expressed as follows: In eliminating a symbol, we may use the formula with whose terminal symbol we are dealing only ‘in the sense afforded it by the introduction of that symbol’.
\end{quote} 
Dummett \cite{Dummett1991logical} developed a philosophical understanding of the normalization results of  Prawitz~\cite{Prawitz1971ideas}, which 
 give a kind of priority to the introduction rules, that yields a notion of valid arguments. The result is P-tV --- see Schroeder-Heister~\cite{Schroeder2006validity} for a succinct explanation. 
 
 More generally, P-tV is about defining a notion of \emph{validity} of objects witnessing that a formula $\phi$ follows by some reasoning from a collection of formulae $\varGamma$. This is quite different from simply giving an interpretation of proofs from some formal system; for example, while the version of P-tV discussed above is closely related to the BHK interpretation of IPL, it is important to distinguish the semantic and computational aspects --- see, for example, Schroeder-Heister~\cite{Schroeder2006validity}.

Meanwhile, B-eS proceeds via a judgement called \emph{support} defined inductively according 
to the structure of formulas with the base case (i.e., the support of 
atoms) given by proof in a base. A \emph{base} is a set of inference rules over atomic propositions, thought of as defining those atoms --- an example is the set of rules above that define `Tammy is a vixen'. Though this approach is closely 
related to possible world semantics in the sense of Beth~\cite{Beth1955} 
and Kripke~\cite{kripke1965semantical} --- see, for example, Goldfarb~\cite{goldfarb2016dummett} and 
Makinson~\cite{makinson2014inferential} --- it remains subtle. 
For example, there are several incompleteness results for intuitionistic logics --- see, for example, Piecha et al.~\cite{Piecha2015failure,Piecha2016completeness,Piecha2019incompleteness}, Goldfarb~\cite{goldfarb2016dummett}, Sandqvist~\cite{Sandqvist2005inferentialist,Sandqvist2009CL,Sandqvist2015hypothesis,Sandqvist2015base}, Stafford~\cite{Stafford2021}. Significantly, a sound and 
complete B-eS for IPL has been given by Sandqvist~\cite{Sandqvist2015base}.   Gheorghiu and Pym~\cite{ptvtobes} have shown that this B-eS captures the declarative content of P-tV.

Sandqvist's B-eS for IPL is the point of departure for this 
paper. Fix a set of atomic propositions $\At$. Given a base $\base{B}$, we write $\proves{\baseB} \at{p}$ to denote that $\at{p} \in \At$ can be derived in $\base{B}$. Support in a base $\baseB$ --- denoted $\suppTor{\baseB}$ --- is defined by the clauses of Figure \ref{fig:Sandqvist:support} in which $\varGamma \neq \emptyset$. We desire to give an analogous semantics for \emph{intuitionistic multiplicative linear logic} (\MILL{}). We study this logic as it is the minimal setting in which we can explore how to set-up B-eS (and P-tS in general) for substructural logics, which enables extension to, for example,  (intuitionistic) Linear Logic \cite{girard1995} and the logic of Bunched Implications \cite{o1999logic}. Again, the aim is not simply to give a proof-theoretic interpretation of IMLL, which already exist, but to define the logical constants in terms of proofs.

\begin{figure}[t]
\hrule \vspace{1mm}
          \[
        \begin{array}{r@{\qquad}l@{\quad}c@{\quad}l}
           \mbox{(At)} & \suppTor{\base{B}} \at{p}  & \text{ iff } &   \proves{\base{B}} \at{p}  \\[1mm]
            (\to) & \suppTor{\base{B}} \phi \to \psi & \text{ iff } & \phi \suppTor{\base{B}} \psi \\[1mm]
            (\land) & \suppTor{\base{B}} \phi \land \psi   & \text{ iff } &   \suppTor{\base{B}} \phi \text{ and }   \suppTor{\base{B}} \psi  \\[1mm]
            (\lor) & \suppTor{\base{B}} \phi \lor \psi & \text{ iff } &  \text{for any } \base{C}  \text{ such that } \base{B} \subseteq \base{C} \text{ and any } \at{p} \in \At, \\ 
            & & & \text{ if } \phi \suppTor{\base{C}} \at{p} \text{ and } \psi \suppTor{\base{C}} \at{p}, \text{ then } \suppTor{\base{C}} \at{p}  \\[1mm]
            (\bot) & \suppTor{\base{B}} \bot & \text{iff} &    \suppTor{\base{B}} \at{p} \text{ for any } \at{p} \in \At \\[1mm]
           \mbox{(Inf)} & \hspace{-1em} \varGamma \suppTor{\base{B}} \phi & \text{ iff } & \text{for any } \base{C} \text{ such that } \base{B} \subseteq \base{C}, \\
           & & & \text{if }   \suppTor{\base{C}} \psi \text{ for any } \psi \in \varGamma, \text{ then } \suppTor{\base{C}} \phi \\[1mm] 
        \end{array}
        \]
\hrule
\vspace{1mm}
    \caption{Sandqvist's Support in a Base}
    \label{fig:Sandqvist:support}
    \vspace{-15pt}
    \labelandtag{cl:tor:at}{(At)}
    \labelandtag{cl:tor:to}{$(\to)$}
    \labelandtag{cl:tor:and}{$(\land)$}
    \labelandtag{cl:tor:or}{$(\lor)$}
    \labelandtag{cl:tor:bot}{$(\bot)$}
    \labelandtag{cl:tor:inf}{(Inf)}
\end{figure}

A compelling reading of \MILL{} is its resource interpretation, which is inherently proof-theoretic --- see Girard \cite{girard1995}.
Accordingly, looking at \ref{cl:tor:inf}, we expect that $\phi$ being supported in a base $\baseB$ relative to some multiset of formulas $\varGamma$ means that the `resources' garnered by $\varGamma$ suffice to produce $\phi$. We may express this by enriching the notion of support with multisets of resources $P$ and $U$ combined with multiset union --- denoted $\,\mcomma\,\,\,$. Then, that the resources garnered by $\varGamma$ are given to $\phi$ is captured by the following property:
\[
\varGamma \suppM{\baseB}{\at{P}} \phi  \qquad \mbox{iff} \qquad \mbox{for any $\baseX \supseteq \baseB$ and any $U$, if $\suppM{\baseX}{\;U} \varGamma$, then 
$\suppM{\baseX}{\;\at{P} \,\mcomma\, \at{U}} \varphi$}
\]
Naively, we may define $\mand$ as a resource-sensitive version of \ref{cl:tor:and}; that is,
\[
    \suppM{\baseB}{\at{P}} \phi \mand \psi \quad \text{ iff } \quad \text{there are } \at{P}_1, \at{P}_2 \text{ such that $\at{P} = (\at{P}_1 \mcomma \at{P}_2)$, $\suppM{\baseB}{\at{P}_1} \phi$, and $\suppM{\baseB}{\at{P}_2} \psi$}
\]
While the semantics is sound, proving completeness is more subtle. We aim to follow the method by Sandqvist \cite{Sandqvist2015hypothesis}, and this clause is not suitable because the following is not the case for \MILL{}: 
\[
 \varGamma \proves{} \phi \mand \psi \quad \text{ iff } \quad \text{there are } \varDelta_1, \varDelta_2 \text{ such that $\varGamma = (\varDelta_1 \mcomma  \varDelta_2)$, $\varDelta_1 \proves{} \phi$, and $\varDelta_2 \proves{} \psi$}
\]
--- a counter-example is the case where $\varGamma$ is the (singleton) multiset consisting of $\phi \otimes \psi$, which denies any non-trivial partition into smaller multisets. We therefore take a more complex clause, which is inspired by the treatment of disjunction in IPL, that enables us to prove completeness using the approach by Sandqvist \cite{Sandqvist2015base}.

There is an obvious difference between the B-eS for IPL and its standard possible world semantics by Kripke \cite{kripke1965semantical} --- namely, the treatment of disjunction ($\lor$) and absurdity ($\bot$). The possible world semantics has the clause,
\[
\mathfrak{M},x \Vdash \phi \lor \psi \qquad \mbox{ iff }\qquad \mathfrak{M},x \Vdash \phi \mbox{ or } \mathfrak{M},x \Vdash \phi
\]
If such a clause is taken in the definition 
of validity in a B-eS for IPL, it leads to incompleteness  --- see, for example, Piecha and Schroeder-Heister ~\cite{Piecha2015failure,Piecha2016completeness}. To yield completeness, Sandqvist \cite{Sandqvist2015hypothesis} uses a more complex form that is close to the elimination rule for disjunction in natural deduction (see Gentzen~\cite{Gentzen} and Prawitz~\cite{prawitz1965}) --- that is, 
\[
{\small 
\begin{array}{r@{\qquad}c@{\qquad}l}
  \suppTor{\base{B}} \phi \lor \psi & \text{ iff } &  \text{for any } \base{C}  \text{ such that } \base{B} \subseteq \base{C} \text{ and any } \at{p} \in \At, \\ 
            & & \text{ if } \phi \suppTor{\base{C}} \at{p} \text{ and } \psi \suppTor{\base{C}} \at{p}, \text{ then } \suppTor{\base{C}} \at{p} 
\end{array}
}
\]
One justification for the clauses is the principle of \emph{definitional reflection} (DR) (see Halln\"as~\cite{hallnas1991partial,hallnas2006proof} and Schroeder-Heister~\cite{schroeder1993rules}):
\begin{quote}
whatever follows from all the premisses of an assertion also follows from the assertion itself
\end{quote}
Taking the perspective that the introduction rules are definitions, DR provides an answer for the way in which the elimination rules follow. Similarly, it justifies that the clauses for the logical constants take the form of their elimination rules. 

Why does the clause for conjunction ($\land$) not take the form given by DR? What DR gives is the \emph{generalized} elimination rule,
\[
\infer{\chi}{\phi \land \psi & \deduce{\chi}{[\phi, \psi]}}
\]
We may modify the B-eS  for IPL by replacing \ref{cl:tor:and} with the following:
\[
\begin{array}{l@{\quad}l@{\quad}c@{\quad}l}
   (\land^\ast) & \suppTor{\base{B}} \phi \land \psi & \text{ iff } & \mbox{for any $\base{C} \baseGeq \base{B}$ and any $\at{p} \in \At$, if $\phi, \psi \suppTor{\baseC} \at{p}$, then $\suppTor{\baseC} \at{p}$}
\end{array}
\]
We show in Section \ref{subsec:revisit-BeS-IPL} that the result does indeed characterize IPL. Indeed, it is easy to see that the generalized elimination rule and usual elimination rule for $\land$ have the same expressive power.

Note, we here take the definitional view of the introduction rules for the logical constants of IPL, and not of bases themselves, thus do not contradict the distinctions made by Piecha and Schroeder-Heister \cite{Schroeder2016atomic,Piecha2017definitional}.

Taking this analysis into consideration, we take the following definition of the multiplicative conjunction that corresponds to the definitional reflection of its introduction rule:
\begin{center}
\begin{tabular}{l@{\qquad}c@{\qquad}l}
$\suppM{\baseB}{\at{P}} \phi \mand \psi$ & \mbox{iff} &\mbox{for any $\baseX \baseGeq \baseB$, resources $\at{U}$, and $\at{p} \in \At$,} \\[1mm]
             & & \mbox{if $\phi \mcomma \psi \suppM{\baseX}{\;\at{U}} \at{p}$, then $\suppM{\baseX}{\;\at{P} \;\mcomma\; \at{U}} \at{p}$} 
 \end{tabular}
 \end{center}
We show in Section \ref{sec:MILL-BeS} that the result does indeed characterize $\MILL{}$.

The paper is structured as follows: in Section \ref{sec:BeS-for-IPL-Sandqvist}, we review the B-eS for IPL given by Sandqvist  \cite{Sandqvist2015base}; in Section \ref{sec:MILL}, we define \MILL{} and provide intuitions about its B-eS; in Section \ref{sec:MILL-BeS}, we formally define the B-eS for \MILL{} and explain its soundness and completeness proofs. The paper ends in Section  \ref{sec:conclusion} with a conclusion and summary of results. 

\section{Base-extension Semantics for IPL} 
\label{sec:BeS-for-IPL-Sandqvist}

In this section, we review the B-eS for IPL given by Sandqvist  \cite{Sandqvist2015base}.  In Section~\ref{subsec:IPL-support-in-base}, we give a terse but complete definition of the B-eS for IPL. In Section~\ref{subsec:IPL-BeS-completeness}, we summarize the completeness proof. Finally, in Section~\ref{subsec:revisit-BeS-IPL}, we discuss a modification of the treatment of conjunction. While IPL is not the focus of this paper, this review provides intuition and motivates the B-eS for \MILL{} in Section~\ref{sec:MILL}. Specifically, the analysis of the treatment of conjunction in IPL motivates the handling of the multiplicative conjunction in \MILL{}. 

Throughout this section, we fix a denumerable set of atomic propositions $\At$, and the following conventions: 
$\at{p}, \at{q}, \dots$ denote atoms; $\at{P}, \at{Q}, \dots$ denote finite sets of atoms; $\varphi, \psi, \theta, \dots$ denote formulas; $\varGamma, \varDelta, \dots$ denote finite sets of formulas.

We forego an introduction to IPL, which is doubtless familiar --- see van Dalen~\cite{vandalen}. For clarity, note that we distinguish sequents $\varGamma \seq \phi$ from judgements $\varGamma \proves{} \phi$ that say that the sequent is valid in IPL.

\subsection{Support in a Base}
\label{subsec:IPL-support-in-base}
The B-eS for IPL begins by defining \emph{derivability in a base}. A (properly) second-level atomic rule --- see Piecha and Schroeder-Heister~\cite{Schroeder2016atomic,Piecha2017definitional} ---  is a natural deduction rule of the following form, in which $q,q_1,...,q_n$ are atoms and $Q_1$,...,$Q_n$ are (possibly empty) sets of atoms: 
\[
 \infer{~~\at{q}~~}{}\qquad 
    \infer{\at{q}}{~~\deduce{\at{q}_{1}}{[Q_{1}]} & ... & \deduce{\at{q}_n}{[Q_n]}~~}  
\]
Importantly, atomic rules are taken \emph{per se} and not closed under substitution. They may be expressed inline as $\left( \at{Q}_1 \seq \at{q}_1, \dots, \at{Q}_n \seq \at{q}_n \right) \Rightarrow \at{q}$ --- note, the axiom case is the special case when the left-hand side is empty, $\Rightarrow \at{q}$. They are read as natural deduction rules in the sense of Gentzen \cite{Gentzen}; thus, $\Rightarrow q$ means that the atom $\at{q}$ may be concluded whenever, while $\left( \at{Q}_1 \seq \at{q}_1, \dots, \at{Q}_n \seq \at{q}_n \right) \Rightarrow \at{q}$ means that one may derive $q$ from a set of atoms $S$ if one has derived  $\at{q}_i$ from $S$ assuming $\at{Q}_i$ for $i=1,...,n$.

A \emph{base} is a set of atomic rules. We write $\baseB, \baseC, \dots$ to denote bases, and $\emptyset$ to denote the empty base (i.e., the base with no rules). We say $\base{C}$ is an \emph{extension} of $\base{B}$ if $\baseC$ is a superset of $\baseB$, denoted $\base{C} \supseteq \base{B}$. 

\begin{definition}[Derivability in a Base]  \label{def:derivability-base-IPL}
\emph{Derivability in a base} $\baseB$ is the least relation $\proves{\baseB}$ satisfying the following:
\begin{description}
    \item[Ref-IPL\label{eq:derive-IPL-ref}] $\at{S}, \at{q} \deriveBaseIPL{\baseB} \at{q}$. 
    \item[App-IPL\label{eq:derive-IPL-app}] If atomic rule $\left( \at{Q}_1 \seq \at{q}_1, \dots, \at{Q}_n \seq \at{q}_n \right) \Rightarrow \at{q}$ is in $\baseB$, and $\at{S}, \at{Q}_i \deriveBaseIPL{\baseB} \at{q}_i$ for all $i = 1, \dots, n$, then $\at{S} \deriveBaseIPL{\baseB} \at{q}$.  
\end{description}
\end{definition}
This forms the base case of the B-eS for IPL: 
\begin{definition}[Sandqvist's Support in a Base] \label{def:BeS-Tor}
\emph{Sandqvist's support in a base $\base{B}$} is the least relation $\suppTor{\base{B}}$ defined by the clauses of Figure~\ref{fig:Sandqvist:support}. A sequent $\varGamma \seq \phi$ is \emph{valid} --- denoted $\varGamma \suppTor{} \phi$ --- iff it is supported in every base, 
\[
    \varGamma \suppTor{} \phi \qquad \mbox{iff} \qquad \mbox{$\varGamma \suppTor{\base{B}} \phi$ holds for any $\base{B}$}
\]
\end{definition}
Every base is an extension of the empty base ($\emptyset$), therefore $\varGamma \suppTor{} \phi$ iff $\varGamma \suppTor{\emptyset} \phi$. Sandqvist \cite{Sandqvist2015base} showed that this semantics characterizes IPL: 
\begin{theorem}[Sandqvist~\cite{Sandqvist2015base}] \label{thm:Sandqvist}
$\varGamma \proves{} \phi$ iff $\varGamma \suppTor{} \phi$
\end{theorem}

Soundness --- that is, $\varGamma\proves{}\phi$ implies $\varGamma \suppTor{} \phi$ --- follows from showing that $\suppTor{}$ respects the rules of Gentzen's \cite{Gentzen} \calculus{NJ}; for example, $\varGamma \suppTor{} \phi$ and $\varDelta \suppTor{} \psi$ implies $\varGamma, \varDelta \suppTor{} \phi \land \psi$. Completeness --- that is, $\varGamma \suppTor{} \phi$ implies $\varGamma \proves{} \phi$ --- is more subtle.  We present the argument in Section \ref{subsec:IPL-BeS-completeness} as it motivates the work in Section~\ref{subsec:mill-BeS-completeness}.

\subsection{Completeness of IPL}
\label{subsec:IPL-BeS-completeness}

We require to show that  $\varGamma \suppTor{} \phi$ implies that there is an $\calculus{NJ}$-proof witnessing $\varGamma \proves{} \phi$. To this end, we associate to each sub-formula $\rho$ of $\varGamma \cup \{\phi\}$ a unique atom $\at{r}$, and construct a base $\base{N}$ such that $\at{r}$ behaves in $\base{N}$ as $\rho$ behaves in $\calculus{NJ}$. Moreover, formulas and their atomizations are semantically equivalent in any extension of $\base{N}$ so that support in $\base{N}$ characterizes both validity 
 and provability. When $\rho \in \At$, we take $\at{r} := \rho$, but for complex $\rho$ we choose $r$ to be alien to $\varGamma$ and $\phi$. 

 \begin{example}
    Suppose $\rho:=\at{p} \land \at{q}$ is a sub-formula of $\varGamma \cup \{\phi\}$. Associate to it a fresh atom $\at{r}$. Since the principal connective of $\rho$ is $\land$, we require $\base{N}$ to contain the following rules:
\[
\infer{\,\,\at{r}\,\,}{\,\,\at{p} & \at{q} \,\,} \qquad \infer{\,\,\at{p}\,\,}{\,\,\at{r}\,\,} \quad \infer{\,\,\at{q}\,\,}{\,\,\at{r}\,\,}
\]
We may write $(\at{p} \land \at{q})^\flat$ for $\at{r}$ so that these rules may be expressed as follows:
\[
    \infer{\,(\at{p}\land\at{q})^\flat\,}{\,\at{p} & \at{q}\,} 
    \qquad \infer{\,\at{p}\,}{\,(\at{p}\land\at{q})^\flat\,} 
    \quad \infer{\,\at{q}\,}{\,(\at{p}\land\at{q})^\flat\,}  \vspace{-2em}
\]
 \end{example}

Formally, given a judgement $\varGamma \suppTor{} \phi$, to every sub-formula $\rho$ associate a unique atomic proposition $\rho^\flat$ as follows:
\begin{itemize}
    \item[-] if $\rho \not \in \At$, then $\rho^\flat$ is an atom that does not occur in any formula in $\varGamma\cup\{\phi\}$;
    \item[-] if $\rho \in \At$, then $\rho^\flat = \rho$.
\end{itemize} 
By \emph{unique} we mean that $(\cdot)^\flat$ is injective --- that is, if $\rho \neq \sigma$, then $\rho^\flat \neq \sigma^\flat$. The left-inverse of $(\cdot)^\flat$ is $(\cdot)^\natural$, and the domain may be extended to the entirety of $\At$ by identity on atoms not in the codomain of $(\cdot)^\flat$. Both functions act on sets pointwise --- that is, $ 
\Sigma^\flat := \{\phi^\flat \mid \phi \in \Sigma \}$ and $\at{P}^\natural := \{\at{p}^\natural \mid \at{p} \in \at{P} \}$. Relative to $(\cdot)^\flat$, let $\base{N}$ be the base containing the rules of Figure~\ref{fig:baseN} for any sub-formulas $\rho$ and $\sigma$ of $\varGamma$ and $\phi$, and any $\at{p} \in \At$. 

\begin{figure}[t]
\hrule
   \vspace{1mm}
   \[
   \begin{array}{c}
       \infer[\irn{\land}^\flat]{(\rho \land \sigma)^\flat}{\rho^\flat & \sigma^\flat} \qquad \infer[\ern{\land}^\flat]{\rho^\flat}{(\rho \land \sigma)^\flat} \quad \infer[\ern{\land}^\flat]{\sigma^\flat}{(\rho \land \sigma)^\flat}
     \quad
      \infer[\ern{\to}^\flat]{\sigma^\flat}{\rho^\flat & (\rho \to \sigma)^\flat} 
   \\[2mm]
    \infer[\irn{\lor}^\flat]{(\rho \lor \sigma)^\flat}{\rho^\flat} 
    \quad 
    \infer[\irn{\lor}^\flat]{(\rho \lor \sigma)^\flat}{\sigma^\flat} 
    \quad
    \infer[\ern{\lor}^\flat]{\at{p}}{(\rho \lor \sigma)^\flat & \deduce{\at{p}}{[\rho^\flat]} & \deduce{\at{p}}{[\sigma^\flat]}}    \quad
     \infer[\irn{\to}^\flat]{(\rho \to \sigma)^\flat}{\deduce{\sigma^\flat}{[\rho^\flat]}}
     \quad
      \infer[\rn{EFQ}^\flat]{\at{p}}{\bot^\flat}
        \end{array}
    \]
       \hrule
   \vspace{1mm}
    \caption{Atomic System $\base{N}$}
    \label{fig:baseN}
    \vspace{-15pt}
\end{figure}

Sandqvist~\cite{Sandqvist2015base} establishes three claims that deliver completeness: 
\begin{description}
    \item[IPL-AtComp\label{lem:Tor:basiccompleteness}] Let $\at{S} \subseteq \At$ and $\at{p} \in \At$ and let $\base{B}$ be a base: 
    $\at{S} \suppTor{\base{B}} \at{p}  \text{ iff } \at{S} \proves{\base{B}} \at{p}. 
    $
    \item[IPL-Flat\label{lem:Tor:flatequivalence}]  For any sub-formula $\xi$ of $\varGamma \cup \{\phi\}$ and $\base{N}' \supseteq \base{N}$: $   \suppTor{\base{N}'} \xi^\flat \text{ iff }  \suppTor{\base{N}'} \xi
    $. 
    \item[IPL-Nat\label{lem:Tor:sharpening}] Let $\at{S} \subseteq \At$ and $\at{p} \in \At$: 
    if $\at{S} \proves{\base{N}} p \text{, then } \at{S}^\natural \proves{} p^\natural$.
\end{description}
The first claim is completeness in the atomic case. The second claim is that $\xi^\flat$ and $\xi$ are equivalent in $\base{N}$ --- that is, $\xi^\flat \suppTor{\base{N}} \xi$ and $\xi \suppTor{\base{N}} \xi^\flat$. Consequently,
\[
\varGamma^\flat \suppTor{\base{N}'} \varphi^\flat \qquad  \text{ iff } \qquad  \varGamma \suppTor{\base{N}'} \varphi 
\]
The third claim is the simulation statement which allows us to make the final move from derivability in $\base{N}$ to derivability in $\calculus{NJ}$. 

\begin{proof}[Theorem~\ref{thm:Sandqvist} --- Completeness.] Assume $\varGamma \suppTor{} \phi$ and let $\base{N}$ be its bespoke base. By \eqref{lem:Tor:flatequivalence}, $\varGamma^\flat \suppTor{\base{N}} \varphi^\flat$. Hence, by \eqref{lem:Tor:basiccompleteness}, $\varGamma^\flat \deriveBaseIPL{\base{N}} \varphi^\flat$. Whence, by \eqref{lem:Tor:sharpening}, $(\varGamma^{\flat})^\natural \proves{} (\varphi^{\flat})^\natural$ --- i.e., $\varGamma \proves{} \varphi$ --- as required. 
\end{proof} 

\subsection{Base-extension Semantics for IPL, revisited}
\label{subsec:revisit-BeS-IPL}
%
Goldfarb~\cite{goldfarb2016dummett,Piecha2019incompleteness} has also given a (complete) proof-theoretic semantics for IPL, but it mimics Kripke's~\cite{kripke1965semantical} semantics. What is interesting about the B-eS in Sandqvist~\cite{Sandqvist2015base} is the way in which it is \emph{not} a representation of the possible world semantics. This is most clearly seen in \ref{cl:tor:or}, which takes the form of the `second-order' definition of disjunction --- that is, 
\[
    U+V = \forall X \left((U \to X)\to(U \to X)\to X\right)
\]
(see Girard~\cite{Girard1989Proofs} and Negri~\cite{troelstra2000basic}). This adumbrates the categorical perspective on B-eS given by Pym et al.~\cite{Pym2022catpts}. Proof-theoretically, the clause recalls the elimination rule for the connective restricted to atomic conclusions, 
\[
    \infer{\at{p}}{
        \phi \lor \psi & 
        \deduce{\at{p}}{[\phi]} & 
        \deduce{\at{p}}{[\psi]}
    }
\]
Dummett~\cite{Dummett1991logical} has shown that such restriction in $\calculus{NJ}$ is without loss of expressive power. Indeed, \emph{all} of the clauses in Figure~\ref{fig:Sandqvist:support} may be regarded as taking the form of the corresponding elimination rules. 

The principle of \emph{definitional reflection}, as described in Section~\ref{sec:introduction} justifies this phenomenon. According to this principle, an alternative candidate clause for conjunction is as follows: 
\labelandtag{cl:tor:andnew}{$(\land^\ast)$}
\[
\begin{array}{r@{\quad}l@{\quad}c@{\quad}l}
   (\land^\ast) & \suppNew{\base{B}} \phi \land \psi & \text{ iff } & \mbox{for any $\base{C} \baseGeq \base{B}$ and any $\at{p} \in \At$,} \mbox{ if $\phi, \psi \suppNew{\baseC} \at{p}$, then $\suppNew{\baseC} \at{p}$}
\end{array}
\]
\vspace{-18pt}
\begin{definition}
\label{def:BeS-new}
The relation $\suppNew{\baseB}$ is defined by the clauses of Figure~\ref{fig:Sandqvist:support} with~\ref{cl:tor:andnew} in place of~\ref{cl:tor:and}. The judgement $\varGamma \suppNew{} \phi$ obtains iff $\varGamma \suppNew{\baseB} \phi$  for any $\baseB$.  
\end{definition}
The resulting semantics is sound and complete for IPL: 
\begin{theorem}
\label{thm:new}
   $\varGamma \suppNew{} \phi$ iff  $\varGamma \proves{} \phi$.
\end{theorem}
\begin{proof}
    We assume the following: for arbitrary base $\baseB$, and formulas $\phi, \psi, \chi$, 
    \begin{description}
        \item[IPL$^\ast$-Monotone] \labelandtag{lem:Tor:monotonenew}{IPL$^\ast$-Monotone} If $\suppNew{\baseB} \phi$, then $\suppNew{\baseC} \phi$ for any $\baseC \supseteq \baseB$. 
        \item[IPL$^\ast$-AndCut] \labelandtag{lem:Tor:key}{IPL$^\ast$-AndCut} If $\suppNew{\baseB} \phi \land \psi$ and $\phi, \psi \suppNew{\baseB} \chi$, then  $\suppNew{\baseB} \chi$. 
    \end{description}
   The first claim follows easily from~\ref{cl:tor:inf}. The second is a generalization of~\ref{cl:tor:andnew}; it follows by induction on the structure of $\chi$ --- an analogous treatment of disjunction was given by Sandqvist~\cite{Sandqvist2015base}.

    By Theorem~\ref{thm:Sandqvist}, it suffices to show that $\varGamma \suppNew{} \varphi$ iff  $\varGamma \suppTor{} \varphi$. For this it suffices to show $ \suppNew{\baseB} \theta$ iff $\suppTor{\baseB} \theta$ for arbitrary $\baseB$ and $\theta$. We proceed by induction on the structure of $\theta$. Since the two relations are defined identically except in the case when the $\theta$ is a conjunction, we restrict attention to this case. 

    First, we show $\suppTor{\baseB} \theta_1 \land \theta_2$ implies $\suppNew{\baseB} \theta_1 \land \theta_2$. By \ref{cl:tor:andnew}, the conclusion is equivalent to the following: for any $\baseC \supseteq \baseB$ and $\at{p} \in \At$, if $\theta_1, \theta_2 \suppNew{\baseC} \at{p}$, then $\suppNew{\baseC} \at{p}$. Therefore, fix $\baseC \supseteq \baseB$ and $\at{p} \in \At$  such that $\theta_1, \theta_2 \suppNew{\baseC} \at{p}$. By~\ref{cl:tor:inf}, this entails the following: if $\suppNew{\baseC} \theta_1$ and $\suppNew{\baseC} \theta_2$, then $\suppNew{\baseC} \at{p}$. By~\ref{cl:tor:and} on the assumption (i.e., $\suppTor{\baseB} \theta_1 \land \theta_2$), we obtain $\suppTor{\baseB} \theta_1$ and $ \suppTor{\baseB} \theta_2$. Hence, by the induction hypothesis (IH), $\suppNew{\baseB} \theta_1$ and $\suppNew{\baseB} \theta_2$. Whence, by \eqref{lem:Tor:monotonenew}, $\suppNew{\baseC} \theta_1$ and $\suppNew{\baseC} \theta_2$. Therefore, $\suppNew{\baseC} \at{p}$. We have thus shown $\suppNew{\baseB} \theta_1 \land \theta_2$, as required.  
   
     Second, we show $\suppNew{\baseB} \theta_1 \land \theta_2$ implies $\suppTor{\baseB} \theta_1 \land \theta_2$. It is easy to see that $\theta_1 , \theta_2 \suppNew{\baseB} \theta_i$ obtains for $i = 1, 2$. Applying \eqref{lem:Tor:key} (setting $\phi=\theta_1$, $\psi= \theta_2$) once with $\chi = \theta_1$ and once  with $\chi=\theta_2$ yields $\suppNew{\baseB} \theta_1$ and $\suppNew{\baseB} \theta_2$. By the IH, $\suppTor{\baseB} \theta_1$ and $\suppTor{\baseB}\theta_2$. Hence, $\suppTor{\baseB} \theta_1 \land \theta_2$, as required. 
\end{proof}

A curious feature of the new semantics is that the meaning of the context-former (i.e., the comma) is not interpreted as $\land$; that is, defining the context-former as
\[
\suppNew{\baseB} \varGamma, \varDelta  \qquad \mbox{iff} \qquad \suppNew{\baseB} \varGamma \mbox{ and }  \suppNew{\baseB} \varDelta
\]
we may express \ref{cl:tor:inf}
\[
\varGamma \suppNew{\baseB} \phi \qquad  \mbox{iff} \qquad \mbox{for any $\base{C} \supseteq \baseB$, if $\suppNew{\baseC} \varGamma$, then $\suppNew{\baseC} \phi$}
\]
The clause for contexts is not the same as the clause for $\land$ in the new semantics. Nonetheless, as shown in the proof of Theorem \ref{thm:new}, they are equivalent at every base --- that is, $\suppNew{\baseB} \phi , \psi$ iff $\suppNew{\baseB} \phi \land \psi$ for any $\baseB$. 

This equivalence of the two semantics yields the following:
\begin{restatable}{corollary}{corSuppFormulaMeanInfAllAtoms}
\label{cor:supp-formula-mean-inf-all-atoms}
    For arbitrary base $\baseB$ and formula $\phi$, $\suppTor{\baseB} \phi$ iff, for any $\baseX \baseGeq \baseB$ and every atom $\at{p}$, if $\phi \suppTor{\baseX} \at{p}$, then $\suppTor{\baseX} \at{p}$. 
\end{restatable}

The significance of this result is that we see that formulas in the B-eS are precisely characterized by their support of atoms.

\section{Intuitionistic Multiplicative Linear Logic}
\label{sec:MILL}

Having reviewed the B-eS for IPL, we turn now to \emph{intuitionistic multiplicative linear logic} (\MILL{}). We first define the logic and then consider the challenges of giving a B-eS for it. This motivates the technical work in Section \ref{sec:MILL-BeS}. Henceforth, we abandon the notation of the previous section as we do not need it and may recycle symbols and conventions.

Fix a countably infinite set $\At$ of atoms. 
\begin{definition}[Formula]
    The set of formulas $( \MILLformula)$ is defined by the following grammar:
    \[
       \phi, \psi ::= \at{p} \in \At \mid \phi \mand \psi \mid \mtop \mid \phi \mto \psi 
    \]
\end{definition}

We use $\at{p}, \at{q}, \dots$ for atoms and $\phi, \psi, \chi, \dots$ for formulas. In contrast to the work on IPL, collections of formulas in \MILL{} are more typically \emph{multisets}. We use $\at{P}, \at{Q}, \dots$ for \emph{finite multisets} of atoms, and $\varGamma, \varDelta, \ldots$ to denote \emph{finite multisets} of formulas. 

We use $\makeMultiset{\,\cdot\,}$ to specify a multiset; for example, $\makeMultiset{\phi, \phi, \psi}$ denotes the multiset consisting of two occurrence of $\phi$ and one occurrences of $\psi$.
The empty multiset  (i.e., the multiset with no members) is denoted $\emptymultiset$. The union of two multisets $\varGamma$ and $\varDelta$ is denoted 
$\varGamma \mcomma \varDelta$. We may identify a multiset containing one element with the element itself; thus, we may write $\psi \mcomma \varDelta$ instead of $\makeMultiset{\psi} \mcomma \varDelta$ to denote the union of multiset $\varDelta$ and the singleton multiset $\makeMultiset{ \psi }$. Thus, when no confusion arises, we may write $\phi_1 \mcomma \dots \mcomma \phi_n$ to denote $[\phi_1,...,\phi_n]$.

\begin{definition}[Sequent]
    A sequent is a pair $\varGamma \seq \phi$ in which $\varGamma$ is a multiset of  formulas and $\phi$ is a formula. 
\end{definition}

We characterize \MILL{} by proof in a natural deduction system. Since it is a substructural logic, we write the system in the format of a sequent calculus as this represents the context management explicitly. We assume general familiarity with sequent calculi --- see, for example, Troelstra and Schwichtenberg~\cite{troelstra2000basic}.

\begin{definition}[System $\calculusMILL{}$]
    The sequential natural deduction system for 
    \MILL{}, denoted 
    $\calculusMILL$, is given by the rules in Figure \ref{fig:MILL-natural-deduction-sequent}. 
\end{definition}

 A sequent $\varGamma \seq \phi$ is a \emph{consequence} of \MILL{} --- denoted $\varGamma \provesMILL \phi$ --- iff there is a \calculusMILL-proof of it.

\begin{figure}[t]
\hrule \vspace{2mm} 
        \labelandtag{eq:mill-axiom}{\rn{ax}}
        \labelandtag{eq:mill-implication-intro}{$\irn{\mto}$}
        \labelandtag{eq:mill-implication-elim}{$\ern{\mto}$}
        \labelandtag{eq:mill-top-intro}{$\irn{\mtop}$}
        \[
            \infer[\rn{ax}]{\phi \seq \phi}{ } 
            \qquad
            \infer[\irn{\mto}]{\varGamma \seq \phi \mto \psi}{\varGamma \mcomma \phi \seq \psi} 
            \qquad 
            \infer[\ern{\mto}]{\varGamma \mcomma \varDelta \seq \psi}{\varGamma \seq \phi \mto \psi & \varDelta \seq \phi} 
            \qquad 
            \infer[\irn{\mtop}]{ \emptyset \seq \mtop }{ }
        \] 
        \labelandtag{eq:mill-top-elim}{$\ern{\mtop}$}
        \labelandtag{eq:mill-conjunction-intro}{$\irn{\mand}$}
        \labelandtag{eq:mill-conjunction-elim}{$\ern{\mand}$}
        \[ 
        \infer[\ern{\mtop}]{\varGamma \mcomma \varDelta \seq \phi}{\varGamma \seq \phi & \varDelta \seq \mtop}
            \qquad 
            \infer[\irn{\mand}]{\varGamma \mcomma \varDelta \seq \phi \mand \psi}{\varGamma \seq \phi & \varDelta \seq \psi}
            \qquad 
            \infer[\ern{\mand}]{\varGamma \mcomma \varDelta \seq \chi}{\varGamma \seq \phi \mand \psi & \varDelta \mcomma \phi \mcomma \psi \seq \chi}
                   \vspace{.5mm} %
        \]
    \hrule 
    \vspace{1mm}
    \caption{The Sequential Natural Deduction System $\calculusMILL$ for IMLL }
    \vspace{-10pt}
    \label{fig:MILL-natural-deduction-sequent}
\end{figure}

One may regard \MILL{} as IPL without the structural rules of weakening and contraction ---  see Do\v{s}en \cite{dosen}. In other words, adding the following rules to $\calculusMILL$ recovers a sequent calculus for IPL:
\[
\infer[\weak]{\varDelta \mcomma \varGamma \seq \phi}{\varGamma \seq \phi} \qquad 
\infer[\cont]{\varDelta \mcomma  \varGamma \seq \phi}{\varDelta \mcomma \varDelta \mcomma \varGamma \seq \phi}
\]

To stay close to the work in Section \ref{sec:BeS-for-IPL-Sandqvist} it is instructive to consider the natural deduction presentation, too. The rule figures may be the same, but their application is not; for example, 
\begin{center}
\begin{tabular}{l@{\qquad}l@{\qquad}l}
    \begin{prooftree}
        \hypo{ \varphi } 
        \hypo{ \psi } 
        \infer2{ \varphi \mand \psi }
    \end{prooftree}
    & \mbox{means}
    & \mbox{if $\varGamma \provesMILL \phi$ and $\varDelta \provesMILL \psi$, then $\varGamma \mcomma \varDelta \provesMILL \phi \otimes \psi$} \\
    & &  \mbox{(i.e., \emph{not} `if $\varGamma \provesMILL \phi$ and $\varGamma \provesMILL \psi$, then $\varGamma \provesMILL \phi \otimes \psi$')}
\end{tabular}
\end{center}
Here, it is important that the context are multisets, not as sets. 

The strict context management in \MILL{} yields the celebrated `resource interpretations' of Linear Logic --- see Girard \cite{girard1995}. The leading example of which is, perhaps, the number-of-uses reading in which a proof of a formula $\phi \mto \psi$ determines a function that \emph{uses} its arguments exactly once. This reading is, however, entirely proof-theoretic and is not expressed in the truth-functional semantics of \MILL{} --- see Girard \cite{girard1995}, Allwein and Dunn \cite{Allwein}, and Coumans et al. \cite{COUMANS201450}. Though these semantics do have sense of `resource' it is not via the number-of-uses reading, but instead denotational in the sense of the treatment of resources in the truth-functional semantics of the logic of Bunched Implications \cite{o1999logic}. The number-of-uses reading is, however, reflected in the categorical semantics  --- see Seely \cite{seely1989linear} and Biermann \cite{Bierman1994intuitionistic,Bierman1995}. 

How do we render support sensitive to the resource reading? The subtlety is that  for $\varGamma \suppTor{} \phi$ (where $\varGamma \neq \emptyset$), we must somehow transmit the resources captured by $\varGamma$ to $\phi$. From Corollary \ref{cor:supp-formula-mean-inf-all-atoms}, we see that in B-eS the content of a formula is captured by the atoms it supports. Therefore, we enrich the support relation with an multiset of atoms $P$,
\[
\begin{array}{l@{\quad}c@{\quad}l}
\varGamma \suppM{\baseB}{P} \phi & \mbox{iff} & 
        \mbox{for any $\baseX \baseGeq \baseB$ and any  $\at{U}$, if $\suppM{\baseX}{\at{U}} \varGamma$, then $\suppM{\baseX}{\at{P} \mcomma \, \at{U}} \phi$}
     \end{array}
\]
where
\[
\begin{array}{l@{\quad}c@{\quad}l}
\suppM{\baseB}{U} \varGamma_1 \mcomma \varGamma_2 & \mbox{iff} & \mbox{there are $U_1$ and $U_2$ such that $U = (U_1 \mcomma U_2)$, $\suppM{\baseX}{\at{U}_1} \varGamma_1$, and  $\suppM{\baseX}{\at{U}_2} \varGamma_2$}
\end{array}
\]
This completes the background on \MILL{}.

\section{Base-extension Semantics for \MILL{}} \label{sec:MILL-BeS}

In this section, we give a B-eS for \MILL{}. It is structured as follows: first, we define support in a base in Section~\ref{subsec:MILL-atomic-system}; second,  we prove soundness in Section~\ref{subsec:mill-BeS-soundness}; finally, we prove completeness in  Section~\ref{subsec:mill-BeS-completeness}.

\subsection{Support in a Base}
\label{subsec:MILL-atomic-system} 
The definition of the B-eS proceeds in line with that for IPL (Section~\ref{sec:BeS-for-IPL-Sandqvist}) while taking substructurality into consideration.
\begin{definition}[Atomic Sequent]
    An \emph{atomic sequent} is a pair $P \seq \at{p}$ in which $P$ is a multiset of atoms and $\at{q}$ is an atom. 
\end{definition}
\begin{definition}[Atomic Rule]
\label{def:MILL-atomic-rule}    
An \emph{atomic rule} is a pair $\mathcal{P} \Rightarrow \at{p}$ in which $\mathcal{P}$ is a (possibly empty) finite set of atomic sequents and $\at{p}$ in an atom.
\end{definition}
\begin{definition}[Base]
A \emph{base} $\baseB$ is a (possibly infinite) set of atomic rules. 
\end{definition}
\begin{definition}[Derivability in a Base]
\label{def:derivability-base-MILL}
The relation $\deriveBaseM{\baseB}$ of \emph{derivability in $\baseB$} is the least relation satisfying the following: 
\begin{description}
    \item[Ref] 
    $\at{p} \deriveBaseM{\baseB} \at{p}$ \labelandtag{eq:derive-mill-ref} {\textsc{ref}}
    \item[App]  \labelandtag{eq:derive-mill-app} {\textsc{app}}
    If $\at{S}_i \mcomma \at{P}_i \deriveBaseM{\baseB} \at{p}_i$ for $i = 1, \dots, n$ and $\left( \at{P}_1 \seq \at{p}_1, \dots, \at{P}_n \seq \at{p}_n \right) \Rightarrow \at{p} \in \baseB$, then $\at{S}_1 \mcomma \dots \mcomma \at{S}_n \deriveBaseM{\baseB} \at{p}$. 
\end{description}
\end{definition}
Note the differences between Definition \ref{def:derivability-base-IPL} and Definition \ref{def:derivability-base-MILL}:  first, in \eqref{eq:derive-mill-ref}, no redundant atoms are allowed to appear, while in \eqref{eq:derive-IPL-ref} they may; second, in \eqref{eq:derive-mill-app}, the multisets $\at{S}_1$,...,$\at{S}_n$ are collected together as a multiset, while in \eqref{eq:derive-IPL-app}, there is one set. 
These differences reflect the fact in the multiplicative setting that `resources' can neither be discharged nor shared. 

\begin{definition}[Support]
\label{def:MILL-BeS}
    That a sequent $\varGamma \seq \phi$ is \emph{supported in the base $\baseB$ using resources $\at{S}$} --- denoted $\varGamma \suppM{\baseB}{\at{S}} \phi$ --- is defined by the clauses of Figure~\ref{fig:MILL-BeS-def} in which $\varGamma$ and $\varDelta$ are non-empty finite multisets of formulas.  The sequent  $\varGamma \seq \phi$ is \emph{supported using resources $S$} --- denoted  $\varGamma \suppM{}{\at{S}} \phi$ --- iff $\varGamma \suppM{\baseB}{\at{S}} \phi$ for any base $\baseB$.  The sequent $\varGamma \seq \phi$ is \emph{valid} --- denoted $\varGamma \suppM{}{} \phi$ --- iff $\varGamma \seq \phi$ is supported using the empty multiset of resources (i.e., $\varGamma \suppM{}{\emptyset} \phi$).
    %
    \end{definition}
    \begin{figure}[t]
        \hrule \vspace{1mm}
        \begin{tabular}{r@{\qquad}l@{\quad}c@{\quad}l}
            \mbox{(At)\labelandtag{eq:MILL-BeS-At}{At}} & \mbox{$\suppM{\baseB}{\at{P}} \at{p}$} &  \mbox{iff} & \mbox{$\at{P} \deriveBaseM{\baseB} \at{p}$} \\ [1mm]
            \mbox{($\mand$)\labelandtag{eq:MILL-BeS-tensor}{$\mand$}} & $\suppM{\baseB}{\at{P}} \phi \mand \psi$ & \mbox{iff} &\mbox{for any $\baseX \baseGeq \baseB$, multiset of atoms $\at{U}$, and atom $\at{p}$,} \\[1mm]
            & & & \mbox{if $\phi \mcomma \psi \suppM{\baseX}{\at{U}} \at{p}$, then $\suppM{\baseX}{\at{P} \mcomma \, \at{U}} \at{p}$} \\ [1mm]
            \mbox{($\mtop$)\labelandtag{eq:MILL-BeS-I}{$\mtop$}} & \mbox{$\suppM{\baseB}{\at{P}} \mtop$} & \mbox{iff} & \mbox{for any $\baseX \baseGeq \baseB$, multiset of atoms $\at{U}$, and atom $\at{p}$,} \\ 
            & & & \mbox{if $\suppM{\baseX}{\;\at{U}} \at{p}$, then $\suppM{\baseX}{\;\at{P} \;\mcomma\;  \at{U}} \at{p}$} \\ [1mm]
            \mbox{($\mto$)\labelandtag{eq:MILL-BeS-imply}{$\mto$}} & \mbox{$\suppM{\baseB}{\at{P}} \phi \mto \psi$} &  \mbox{iff} & 
            \mbox{$\phi \suppM{\baseB}{\at{P}} \psi$} \\ [1mm]
        (\,\,$\mcomma$\,\,)\labelandtag{eq:MILL-BeS-comma}{\,\,$\mcomma$\,\,} & 
           $ \suppM{\baseB}{P} \varGamma \mcomma \varDelta$ & iff &  there are $U$ and $V$ such that $P = (U \mcomma V)$, $\suppM{\baseB}{\at{U}} \varGamma$, and  $\suppM{\baseB}{V} \varDelta$
            \\[1mm]
             \mbox{(Inf)\labelandtag{eq:MILL-BeS-Inf}{Inf}} 
                 & \hspace{-1.4em} $\varGamma \suppM{\baseB}{P} \phi$ & \mbox{iff} & 
        \mbox{for any $\baseX \baseGeq \baseB$ and any  $\at{U}$, if $\suppM{\baseX}{\at{U}} \varGamma$, then $\suppM{\baseX}{\;\at{P} \;\mcomma\; \at{U}} \phi$}
        \end{tabular}
        \vspace{1mm} \hrule \vspace{1mm}    
        \caption{Base-extension Semantics for \MILL{} }
        \label{fig:MILL-BeS-def}
        \vspace{-15pt}
    \end{figure} 

%

It is easy to see that Figure \ref{fig:MILL-BeS-def} is an inductive definition on a structure of formulas that prioritizes conjunction ($\mand$) over implication ($\mto$) --- an analogous treatment in IPL with disjunction ($\lor$) prioritized over implication ($\to$) has been given by Sandqvist~\cite{Sandqvist2015base}. 
As explained in Section \ref{sec:MILL}, the purpose of the  multisets of atoms $\at{S}$ in the support relation $\suppM{\baseB}{\at{S}}$ is to express the susbtructurality of the logical constants. The naive ways of using multisets of formulas rather than multisets of atoms --- for example, $\varGamma \suppM{\base{B}}{\varDelta} \phi$ iff $\suppM{\base{B}}{\varGamma \, \mcomma \, \varDelta} \phi$ --- results in impredicative definitions of support.

We read \eqref{eq:MILL-BeS-Inf} as saying that $\varGamma \suppM{\baseB}{\at{S}} \varphi$ (for $\varGamma \neq \emptyset$) means, for any extension $\baseX$ of $\baseB$, if $\varGamma$ is supported in $\baseX$ with some resources $\at{U}$ (i.e. $\suppM{\baseX}{\at{U}} \varGamma$), then $\varphi$ is also supported by combining the resources $\at{U}$ with the resources $\at{S}$ (i.e., $\suppM{\baseX}{\;\at{S} \; \mcomma \; \at{U}} \phi$). 

The following observation on the monotonicity of the semantics with regard to base extensions follows immediately by unfolding definitions: 
\begin{restatable}{proposition}{MillBeSMonotonicity}
\label{lem:BeS-Mill-monotone-on-base}
    If $\varGamma \suppM{\baseB}{\at{S}} \phi$ and $\baseC \baseGeq \baseB$, then $\varGamma \suppM{\baseC}{\at{S}} \phi$. 
\end{restatable}
From this proposition we see the following:  $\varGamma \suppM{}{\at{S}} \phi$ iff $\varGamma \suppM{\emptymultiset}{\at{S}} \phi$, and $\varGamma \suppM{}{} \phi$ iff $\varGamma \suppM{\emptyset}{\emptymultiset} \phi$.
    As expected, we do not have monotonicity on resources --- that is, $\varGamma \suppM{}{\at{S}} \phi$ does not, in general, imply $\varGamma \suppM{}{\;\at{S} \; \mcomma \; T} \phi$ for arbitrary $\at{T}$. This exposes the different parts played by bases and the resources in the semantics: bases are the setting in which a formula is supported, resources are tokens used in that setting to establish the support.

A distinguishing aspect of support is the structure of \eqref{eq:MILL-BeS-Inf}. In one direction, it is merely cut, but in the other it says something stronger. The completeness argument will go through the atomic case (analogous to the treatment of IPL in Section \ref{subsec:IPL-BeS-completeness}), and the following proposition suggests that the setup is correct:
\begin{restatable}{proposition}{lemMILLAtomicCut}
\label{lem:mill-atomic-cut}
    The following two propositions are equivalent for arbitrary base $\baseB$, multisets of atoms $\at{P}, \at{S}$, and atom $\at{q}$, where we assume $\at{P} = \makeMultiset{\at{p}_1, \dots, \at{p}_n}$: 
    \begin{enumerate}
        \item $\at{P} \mcomma \at{S} \deriveBaseM{\baseB} \at{q}$. \label{eq:mill-atomic-cut-1}
        \item for any $\baseX \baseGeq \baseB$ and multisets of atoms $\at{T_1}, \dots, \at{T_n}$, if $\at{T}_i \deriveBaseM{\baseX} \at{p}_i$ holds for all $i = 1, \dots,n$, then $\at{T}_1 \mcomma \dots \mcomma \at{T}_n \mcomma \at{S} \deriveBaseM{\baseX} \at{q}$. \label{eq:mill-atomic-cut-2}
    \end{enumerate}
\end{restatable}

It remains to prove soundness and completeness. 

\vspace{-5pt}
\subsection{Soundness}
\label{subsec:mill-BeS-soundness}
\begin{restatable}[Soundness]{theorem}{thmMILLBeSSoundness}
\label{thm:MILL-BeS-soundness}
   If $\varGamma \provesMILL \phi$, then $\varGamma \suppM{}{} \phi$. 
\end{restatable}
%
The argument follows a typical strategy of showing that the semantics respects the rules of $\calculusMILL$ --- 
that is, for any $\varGamma, \varDelta, \phi, \psi,$ and $\chi$:
\labelandtag{eq:mill-soundness-axiom-0}{      \mbox{(Ax)}}
\labelandtag{eq:mill-soundness-implication-intro-0}{\mbox{($\mto$I)}}
\labelandtag{eq:mill-soundness-implication-elim-0}{\mbox{($\mto$E)}}
\labelandtag{eq:mill-soundness-conjunction-intro-0}{\mbox{($\mand$I)}}
\labelandtag{eq:mill-soundness-conjunction-elim-0}{\mbox{($\mand$E)}}
\labelandtag{eq:mill-soundness-top-intro-0}{\mbox{($\mtop$I)}}
\labelandtag{eq:mill-soundness-top-elim-0}{\mbox{($\mtop$E)}}
\[
\hspace{-3cm}
\begin{array}{r@{\quad}l}
      \mbox{(Ax)} &  \phi \suppM{}{} \phi\\
     \mbox{($\mto$I)} &  \mbox{If $\varGamma, \phi \suppM{}{} \psi$, then $\varGamma \suppM{}{} \phi \mto \psi$}\\
     \mbox{($\mto$E)} &  \mbox{If $\varGamma \suppM{}{} \phi \mto \psi$ and $\varDelta \suppM{}{} \phi$, then $\varGamma, \varDelta \suppM{}{} \psi$} \\
    \mbox{($\mand$I)} & \mbox{If $\varGamma \suppM{}{} \phi$ and $\varDelta \suppM{}{} \psi$, then $\varGamma, \varDelta \suppM{}{} \phi \mand \psi$} \\
     \mbox{($\mand$E)} & \mbox{If $\varGamma \suppM{}{} \phi \mand \psi$ and $\varDelta \mcomma \phi \mcomma \psi \suppM{}{} \chi$, then $\varGamma \mcomma \varDelta \suppM{}{} \chi$} \\
    \mbox{($\mtop$I)} & \suppM{}{} \mtop \\
   \mbox{($\mtop$E)}  & \mbox{If $\varGamma \suppM{}{} \chi$ and $\varDelta \suppM{}{} \mtop$, then $\varGamma, \varDelta \suppM{}{} \chi$}
\end{array}
\]

These follow quickly from the fact that the clauses of each connective in Figure \ref{fig:MILL-BeS-def} takes the form of its elimination rules. The only subtle cases are \ref{eq:mill-soundness-conjunction-elim-0} and \ref{eq:mill-soundness-top-elim-0}. 

To show \ref{eq:mill-soundness-top-elim-0}, suppose $\varGamma \suppM{}{} \chi$ and $\varDelta \suppM{}{} \mtop$. We require to show $\varGamma \mcomma \varDelta \suppM{}{} \chi$. By \eqref{eq:MILL-BeS-Inf}, we fix some base $\baseB$ and multisets of atoms $\at{P}$ and $\at{Q}$ such that $\suppM{\baseB}{\at{P}} \varGamma$ and $\suppM{\baseB}{\at{Q}} \varDelta$. It  remains to verify $\suppM{\baseB}{\at{P} \, \mcomma\, \at{Q}} \chi$. When $\chi$ is atomic, this follows immediately from $\suppM{\baseB}{\at{P}} \chi$ and $\suppM{\baseB}{\at{Q}} \mtop$ by \eqref{eq:MILL-BeS-I}. To handle non-atomic $\chi$, we require the following: 
\begin{restatable}{lemma}{lemMILLMtopKeyLemma}
\label{lem:mill-mtop-key-lemma}
For arbitrary base $\baseB$, multisets of atoms $\at{S}, \at{T}$, and formula $\chi$, if 
\begin{enumerate*}
    \item $\suppM{\baseB}{\at{S}} \mtop$,\label{eq:mill-soundness-mtop-key-lemma-condition-1} 
    \item $\suppM{\baseB}{\at{T}} \chi$,\label{eq:mill-soundness-mtop-key-lemma-condition-2}
\end{enumerate*}
then 
\begin{enumerate*}
    \setcounter{enumi}{2}
    \item $\suppM{\baseB}{\at{S} \, \mcomma \, \at{T}} \chi$. \label{eq:mill-soundness-mtop-key-lemma-conclusion}
\end{enumerate*}
\end{restatable}
This lemma follows by induction on the structure of $\chi$, with the base case given by \eqref{eq:MILL-BeS-I}. One cannot use this general form to define $\mtop$ as it would result in an impredicative definition of support.

Similarly, we require the following to prove \ref{eq:mill-soundness-conjunction-elim-0}: 
\begin{restatable}{lemma}{lemMILLConjKeyLemma}
\label{lem:mill-conj-key-lemma}
For arbitrary base $\baseB$, multisets of atoms $\at{S}, \at{T}$, and formulas $\phi, \psi, \chi$, 
if 
\begin{enumerate*}
    \item $\suppM{\baseB}{\at{S}} \phi \mand \psi$,\label{eq:mill-soundness-key-lemma-condition-1} 
    \item $\phi \mcomma \psi \suppM{\baseB}{\at{T}} \chi$,\label{eq:mill-soundness-key-lemma-condition-2}
\end{enumerate*}
then 
\begin{enumerate*}
    \setcounter{enumi}{2}
    \item $\suppM{\baseB}{\at{S} \, \mcomma \, \at{T}} \chi$. \label{eq:mill-soundness-key-lemma-conclusion}
\end{enumerate*}
\end{restatable}

With these results, we may prove soundness:

\begin{proof}[Theorem~\ref{thm:MILL-BeS-soundness} --- sketch]
We demonstrate \ref{eq:mill-soundness-conjunction-intro-0} and \ref{eq:mill-soundness-conjunction-elim-0}. 
 
    \ref{eq:mill-soundness-conjunction-intro-0}. Assume $\varGamma \suppM{}{} \phi$ and $\varDelta \suppM{}{} \psi$. We require to show $\varGamma \mcomma \varDelta \suppM{}{} \phi \mand \psi$. By \eqref{eq:MILL-BeS-Inf}, the conclusion is equivalent to the following: for any base $\baseB$, for any multisets of atoms $\at{T}$ and $\at{S}$ , if $\suppM{\baseB}{T} \varGamma$ and $\suppM{\baseB}{S} \varDelta$, then $\suppM{\baseB}{T \, \mcomma \, S} \phi \mand \psi$. So we fix some $\baseB$ and $\at{T}$, $\at{S}$ such that $\suppM{\baseB}{T} \varGamma$ and $\suppM{\baseB}{S} \varDelta$, and show that $\suppM{\baseB}{ \at{T} \, \mcomma \, \at{S} } \phi \mand \psi$. By \eqref{eq:MILL-BeS-tensor}, it suffices to show, for arbitrary $\baseC \baseGeq \baseB$, multiset of atoms $\at{U}$, and atom $\at{p}$, if $\phi \mcomma \psi \suppM{\baseC}{\at{U}} \at{p}$, then $\suppM{\baseC}{ \at{T} \, \mcomma \, \at{S} \,\mcomma\, \at{U} } \at{p}$. So we fix some $\baseC \baseGeq \baseB$, multiset of atoms $\at{U}$, and atom $\at{p}$ such that $\phi \mcomma \psi \suppM{\baseC}{\at{U}} \at{p}$, and the goal is to show that $\suppM{\baseC}{ \at{T} \, \mcomma \, \at{S} \, \mcomma \, \at{U} } \at{p}$.  From the assumptions $\varGamma \suppM{}{} \varphi$ and $\varDelta \suppM{}{} \psi$, we see that $\suppM{\baseB}{\at{S} \, \mcomma \, \at{T}}  \phi \, \mcomma \, \psi $ obtains. Therefore, by monotonicity, $\suppM{\baseC}{\at{S} \, \mcomma \, \at{T}} \phi\mcomma\psi $ obtains. By \eqref{eq:MILL-BeS-Inf}, this suffices for $\phi \mcomma \psi \suppM{\baseC}{\at{U}} \at{p}$, to yield  $\suppM{\baseC}{T \, \mcomma \, S \, \mcomma \, \at{U}} \at{p}$, as required. 
    
    \ref{eq:mill-soundness-conjunction-elim-0}.  Assume $\varGamma \suppM{}{} \phi \mand \psi$ and $\varDelta \mcomma \phi \mcomma \psi \suppM{}{} \chi$. We require to show $\varGamma \mcomma \varDelta \suppM{}{} \chi$. By \eqref{eq:MILL-BeS-Inf}, it suffices to assume $\suppM{\baseB}{\at{S}} \varGamma$ and $\suppM{\baseB}{\at{T}} \varDelta$ and show that $\suppM{\baseB}{\at{S} \, \mcomma \, \at{T}} \chi$. First, $\varGamma \suppM{}{} \phi \mand \psi$ together with $\suppM{\baseB}{\at{S}} \varGamma$ entails that $\suppM{\baseB}{\at{S}} \phi \mand \psi$. Second, by \eqref{eq:MILL-BeS-Inf}, $\varDelta \mcomma \phi \mcomma \psi \suppM{}{} \chi$ is equivalent to the following:
    \begin{equation*}
    \label{eq:mill-soundness-conj-elim-proof-2-0}
        \text{for any } \baseX \text{ and } \at{P}, \at{Q}, \text{ if } \suppM{\baseX}{\at{P}} \varDelta \text{ and } \suppM{\baseX}{\at{Q}} \phi \mcomma \psi, \text{ then } \suppM{\baseX}{\at{P} \, \mcomma \, \at{Q}} \chi 
    \end{equation*}
    Since $\suppM{\baseB}{\at{T}} \varDelta$, setting $\at{P} := \at{T}$ and $\at{Q} := \at{S}$, yields,
    \begin{equation}
    \label{eq:mill-soundness-conj-elim-proof-1-0}
        \text{for any } \baseX \baseGeq \baseB, \text{ if } \suppM{\baseX}{\at{S}} \phi \mcomma \psi , \text{ then } \suppM{\baseX}{\at{T} \, \mcomma \, \at{S}} \chi 
    \end{equation}
    Now, given $\suppM{\baseB}{\at{S}} \phi \mand \psi$ and \eqref{eq:mill-soundness-conj-elim-proof-1-0}, we can apply Lemma~\ref{lem:mill-conj-key-lemma} and conclude $\suppM{\baseB}{\at{S} \, \mcomma \, \at{T}} \chi$. 
\end{proof}

\subsection{Completeness}
\label{subsec:mill-BeS-completeness}

\begin{restatable}[Completeness]{theorem}{thmMILLBeSCompleteness}
\label{thm:MILL-BeS-completeness}
    If $\varGamma \suppM{}{} \phi$, then $\varGamma \provesMILL \phi$. 
\end{restatable}
The argument follows the strategy used by Sanqvist~\cite{Sandqvist2015base} for IPL --- see Section~\ref{subsec:IPL-BeS-completeness}. We explain the main steps.

Let $\Xi$ be the set of all sub-formulas of $\varGamma \cup \{\phi\}$. Let $\flatmill{(\cdot)} \colon \Xi \to \At$ be an injection that is fixed on $\Xi\cap\At$ --- that is, $\flatmill{\at{p}} = \at{p}$ for $\at{p}\in \Xi \cap \At$. Let  $\deflatmill{(\cdot)}$ be the left-inverse of $\flatmill{(\cdot)}$ --- that is $\deflatmill{\at{p}} = \chi$ if $\at{p}=\flatmill{\chi}$, and $\deflatmill{\at{p}} = \at{p}$ if $\at{p}$ is not in the image of $\flatmill{(\cdot)}$. Both act on multisets of formulas pointwise; that is, $\flatmill{\varDelta} \coloneqq \makeMultiset{\flatmill{\delta} \mid \delta \in \varDelta}$ and $\deflatmill{\at{P}} \coloneqq \makeMultiset{\deflatmill{\at{p}} \mid \at{p} \in \at{P}}$.

We construct a base $\baseMill$ such that $\flatmill{\phi}$ behaves in $\baseMill$ as $\phi$ behaves in $\calculusMILL{}$. The base $\baseMill$ contains all instances of the rules of Figure \ref{fig:baseM} when $\sigma$ and $\tau$ range over $\Xi$, 
and $\at{p}$ ranges over $\At$. We illustrate how $\baseMill$ works with an example. 
\begin{figure}[t]
\hrule
\vspace{1mm}
\[
\begin{array}{rl@{\qquad}rl}
        \flatmill{\irn \mto}: & ( \flatmill{\sigma} \seq \flatmill{\tau} ) \Rightarrow \flatmill{(\sigma \mto \tau)}
        & \flatmill{\ern \mto}: & \left( \seq \flatmill{(\sigma \mto \tau)} , \seq \flatmill{\sigma} \right) \Rightarrow \flatmill{\tau}\\[2mm]
       \flatmill{\irn \mand}: & \left( \seq \flatmill{\sigma}, \seq \flatmill{\tau} \right) \Rightarrow \flatmill{(\sigma \mand \tau)}
        & \flatmill{\ern \mand}: &
         \left( \seq \flatmill{ (\sigma \mand \tau) }, \flatmill{\sigma} \mcomma \flatmill{\tau} \seq \at{p} \right) \Rightarrow \at{p} 
        \\[2mm]
        \flatmill{\irn \mtop}: &  \Rightarrow \flatmill{\mtop}
        & \flatmill{\ern \mtop}: & \left( \seq \flatmill{\mtop}, \seq \at{p} \right) \Rightarrow \at{p} 
    \end{array} 
    \]
    \hrule
    \vspace{1mm}
    \caption{Atomic System $\baseMill$ \label{fig:baseM}}
    \vspace{-20pt}
    \end{figure}
    %

\begin{example}
\label{ex:MILL-flat-and-defalt}
Consider the sequent $\varGamma \seq \phi$ where $\varGamma = \makeMultiset{ \at{p}_1 \mcomma \at{p}_2 \mcomma \at{p}_1 \mand \at{p}_2 \mto \at{q}, \at{p_1} }$ and $\phi = \at{q} \mand \at{p_1} $. By definition, $\Xi := \{ \at{p}_1, \at{p}_2, \at{p}_1 \mand \at{p}_2 \mto \at{q}, \at{p}_1 \mand \at{p}_2, \at{q}, \at{q} \mand \at{p}_1 \}$, and, therefore, the image of $\flatmill{(\cdot)}$ is $ \{ \at{p}_1, \at{p}_2, \at{q}, \flatmill{( \at{p}_1 \mand \at{p}_2 \mto \at{q} )}, \flatmill{( \at{p}_1 \mand \at{p}_2 )}, \flatmill{( \at{q} \mand \at{p}_1 )} \}$.

That $\varGamma \provesMILL \phi$ obtains is witnessed by the following $\calculusMILL$-proof: 
\begin{equation*}
\label{eq:mill-flat-eg-proof-tree-1}
    \begin{prooftree}
    \hypo{ } 
    \infer1[\textsf{ax}]{ \at{p}_1 \seq \at{p}_1 }
    \hypo{ } 
    \infer1[\textsf{ax}]{ \at{p}_2 \seq \at{p}_2 }
    \infer2[$\mand_{\mathsf{I}}$]{ \at{p}_1 \mcomma \at{p}_2 \seq \at{p}_1 \mand \at{p}_2 } 
    \hypo{  } 
    \infer1[\textsf{ax}]{ \at{p}_1 \mand \at{p}_2 \mto \at{q} \seq \at{p}_1 \mand \at{p}_2 \mto \at{q} }
    \infer2[$\mto_{\mathsf{E}}$]{ \at{p}_1 \mcomma \at{p}_2 \mcomma \at{p}_1 \mand \at{p}_2 \mto \at{q} \seq \at{q} }
    \hypo{ }
    \infer1[\textsf{ax}]{ \at{p}_1 \seq \at{p}_1 } 
    \infer2[$\mand_{\mathsf{I}}$]{ \at{p}_1 \mcomma \at{p}_2 \mcomma \at{p}_1 \mand \at{p}_2 \mto \at{q} \mcomma \at{p}_1 \seq \at{q} \mand \at{p}_1 }
\end{prooftree}
\end{equation*}

The base $\baseMill$ is designed so that we may simulate the rules of $\calculusMILL$; for example, the $\ern \mand$ is simulated by using \eqref{eq:derive-mill-app} on $\ern{\mand}^\flat$, 
\[
\begin{array}{lcl}
( \emptymultiset \seq \flatmill{( \sigma \mand \tau )}, \flatmill{\sigma} \mcomma \flatmill{\tau} \seq \flatmill{\gamma}) \Rightarrow \flatmill{\gamma} & \mbox{ means } & 
\mbox{if $\flatmill{\varDelta} \deriveBaseM{\baseMill} \flatmill{( \sigma \mand \tau )}$ and $\flatmill{\Sigma} \mcomma \flatmill{\sigma} \mcomma \flatmill{\tau} \deriveBaseM{\baseMill} \flatmill{ \gamma }$} \\ & &
           \mbox{then $\flatmill{\varDelta} \mcomma \flatmill{\Sigma} \deriveBaseM{\baseMill} \flatmill{\gamma} $ }
\end{array}
\]
In this sense, the proof above is simulated by the following steps:
\begin{enumerate}
    \item[(i)] By \eqref{eq:derive-mill-ref}, (1) $\at{p}_1 \deriveBaseM{\baseMill} \at{p}_1$;
    (2) $\at{p}_2 \deriveBaseM{\baseMill} \at{p}_2$; 
    (3) $\flatmill{\left( \at{p}_1 \mand \at{p}_2 \mto \at{q} \right)} \deriveBaseM{\baseMill} \flatmill{\left( \at{p}_1 \mand \at{p}_2 \mto \at{q} \right)}$ 
    \item[(ii)] By \eqref{eq:derive-mill-app}, using $(\irn \mand)$ on (1) and (2), we obtain (4) $\at{p}_1 \mcomma \at{p}_2 \deriveBaseM{\baseMill} \flatmill{(\at{p}_1 \mand \at{p}_2)}$ 
    \item[(iii)] By \eqref{eq:derive-mill-app}, using $\flatmill{(\ern \mto)}$ on (3) and (4), we obtain (5) $\flatmill{\left( \at{p}_1 \mand \at{p}_2 \mto \at{q} \right)} \mcomma
    \at{p}_1 \mcomma \at{p}_2 \deriveBaseM{\baseMill} \at{q}$ 
    \item[(iv)] By \eqref{eq:derive-mill-app}, using $\flatmill{(\irn \mand)}$ on (1) and (5). we have $\flatmill{\left( \at{p}_1 \mand \at{p}_2 \mto \at{q} \right)} \mcomma
    \at{p}_1 \mcomma \at{p}_2 \mcomma \at{p}_1 \deriveBaseM{\baseMill} \flatmill{\left( \at{q} \mand \at{p}_1 \right)}$. 
\end{enumerate}
Significantly, steps (i)--(iv) are analogues of the steps in the proof tree above.
\end{example}

Theorem \ref{thm:MILL-BeS-completeness} (Completeness) follows from the following three observations, which are counterparts to \eqref{lem:Tor:basiccompleteness}, \eqref{lem:Tor:flatequivalence}, and \eqref{lem:Tor:sharpening} from Section~\ref{subsec:IPL-BeS-completeness}: 
\begin{description}
    \item[IMLL-AtComp]\labelandtag{eq:mill-atomic-sound-and-complete}{IMLL-AtComp} 
    For any $\baseB$, $\at{P}$, $\at{S}$, and $\at{q}$, $\at{P} \mcomma \at{S} \deriveBaseM{\baseB} \at{q}$ iff $\at{P} \suppM{\baseB}{\at{S}} \at{q}$. 
    \item[IMLL-Flat\label{eq:mill-complete-flat-iff-0}] For any $\xi \in \Xi$, $\baseX \baseGeq \baseMill$ and $\at{U}$, $\suppM{\baseX}{\at{U}} \flatmill{\xi}$ iff $\suppM{\baseX}{\at{U}} \xi$.  
    \item[IMLL-Nat\label{eq:mill-complete-deflat-0}] For any 
    $\at{P}$ and 
    $\at{q}$, if $\at{P} \deriveBaseM{\baseMill} \at{q}$ then $\deflatmill{\at{P}} \provesMILL \deflatmill{\at{q}}$.  
\end{description}

\eqref{eq:mill-atomic-sound-and-complete} follows from Proposition \ref{lem:mill-atomic-cut} and is the base case of completeness. \eqref{eq:mill-complete-flat-iff-0} formalizes the idea that every formula $\xi$ appearing in $\varGamma \seq \varphi$ behaves the same as $\flatmill{\xi}$ in any base extending $\baseMill$. Consequently, 
$\varGamma^\flat \suppTor{\baseMill} \phi^\flat$ iff $\varGamma \suppTor{\baseMill} \phi$. 
\eqref{eq:mill-complete-deflat-0} intuitively says that $\baseMill$ is a faithful atomic encoding of $\calculusMILL$, witnessed by $\deflatmill{(\cdot)}$. This together with \eqref{eq:mill-complete-flat-iff-0} guarantee that every $\xi \in \Xi$ behaves in $\baseMill$ as $\flatmill{\xi}$ in $\baseMill$, thus as $\deflatmill{\left( \flatmill{\xi} \right)} = \xi$ in $\calculusMILL$. 
%
\begin{proof}[Theorem~\ref{thm:MILL-BeS-completeness} --- Completeness] Assume $\varGamma \suppM{}{} \phi$ and let $\base{M}$ be the bespoke base for $\varGamma \seq \phi$. By \eqref{eq:mill-complete-flat-iff-0}, $\varGamma^\flat \suppM{\base{M}}{\emptyset} \phi^\flat$. Therefore, by \eqref{eq:mill-atomic-sound-and-complete}, we have $\flatmill{\varGamma} \proves{\baseMill} \flatmill{\phi}$. Finally, by \eqref{eq:mill-complete-deflat-0}, $\deflatmill{\left( \flatmill{\varGamma} \right)} \provesMILL \deflatmill{\left( \flatmill{\phi} \right)}$, namely $\varGamma \proves{} \phi$. 
\end{proof}

%
  
%

\section{Conclusion} \label{sec:conclusion}
Proof-theoretic semantics (P-tS) is the paradigm of meaning in logic based on proof, as opposed to truth. A particular form of P-tS is \emph{base-extension semantics} (B-eS) in which one defines the logical constants by means of a \emph{support} relation indexed by a base --- a system of natural deduction for atomic propositions --- which grounds the meaning of atoms by proof in that base.  This paper provides a sound and complete base-extension semantics for \emph{intuitionistic multiplicative linear logic} (\MILL{}). 

The B-eS for IPL given by Sandqvist \cite{Sandqvist2015base} provides a strategy for the problem. The paper begins with a brief but instructive analysis of this work that reveals \emph{definitional reflection} (DR) as an underlying principle delivering the semantics; accordingly, in Section \ref{subsec:revisit-BeS-IPL}, the paper modifies the B-eS for IPL to strictly adhere to DR and proves soundness and completeness of the result. Moreover, the analysis highlights that essential to B-eS is a transmission of proof-theoretic content: a formula $\phi$ is supported in a base $\baseB$ relative to a context $\varGamma$ iff, for any extension $\baseC$ of $\baseB$, the formula $\phi$ is supported in $\baseC$ whenever $\varGamma$ is supported in $\baseC$.

 With this understanding of B-eS of IPL, the paper gives a `resource-sensitive' adaptation by enriching the support relation to carry a multiset of atomic `resources' that enable the transmission of proof-theoretic content. This captures the celebrated `resource reading' of \MILL{} which is entirely proof-theoretic --- see Girard \cite{girard1995}. The clauses of the logical constants are then delivered by DR on their introduction rules. Having set up the B-eS for \MILL{} in this principled way, soundness and completeness follow symmetrically to the preceding treatment of IPL. 

To date, P-tS has largely been restricted to classical and intuitionistic propositional logics. This paper provides the first step toward a broader analysis. In particular, the analysis in this paper suggests a general methodology for delivering B-eS for other substructural logics such as, \emph{inter alia}, (intuitionistic) Linear Logic~\cite{girard1995} (LL) and the logic of Bunched Implications~\cite{o1999logic} (BI). While it is straightforward to add the additive connectives of LL, with the evident semantic clauses following IPL and with the evident additional cases in the proofs, it is less apparent how to handle the exponentials. For BI, the primary challenge is to appropriately account for the \emph{bunched} structure of contexts, and to enable and confine weakening and contraction to the additive context-former.  

Developing the P-tS for substructural logics is 
 valuable because of their deployment in the verification and modelling of systems. Significantly, P-tS has shown the be useful in simulation modelling --- see, for example, Kuorikoski and Reijula \cite{jaakko}. Of course, more generally, we may ask what conditions a logic must satisfy in order to provide a B-eS for it.

 \section*{Acknowledgements}
 We are grateful to Yll Buzoku, Diana Costa, Sonia Marin, and Elaine Pimentel for many discussion on the P-tS for substructural logics, and to Jonte Deakin for his careful reading of and feedback on an earlier draft of this article.  Similarly, we would like to thank the reviewers for their helpful comments and remarks.

\bibliographystyle{splncs04}
\bibliography{refs}

\appendix

\section{Omitted proofs from Section~\ref{subsec:revisit-BeS-IPL}} 
\label{app:IPL-Revisted}
The following contains proofs for the claims \ref{lem:Tor:monotonenew} and \ref{lem:Tor:key}  in the proof of Theorem~\ref{thm:new}.

\begin{lemma}[\ref{lem:Tor:monotonenew}]
    If $\Gamma \suppNew{\baseB} \phi$, then $\Gamma \suppNew{\baseC} \varphi$ for any $\baseC \supseteq \baseB$. 
\end{lemma}
\begin{proof}
    By \ref{cl:tor:inf}, the conclusion $\Gamma \suppNew{\baseC} \phi$ means: for every $\baseD \baseGeq \baseC$, if $\suppNew{\baseD} \gamma$ for every $\gamma \in \Gamma$, then $\suppNew{\baseD} \phi$.  
   Since $\baseD \baseGeq \baseC \baseGeq \baseB$, this follows by  \ref{cl:tor:inf} on the hypothesis $\Gamma \suppNew{\baseB} \phi$.
\end{proof}
\begin{lemma}[\ref{lem:Tor:key}] 
    If $\suppNew{\baseB} \phi \land \psi$ and $\phi , \psi \suppNew{\baseB} \chi$, then $\suppNew{\baseB} \chi$.
\end{lemma}
\begin{proof}
       We proceed by induction on the structure of $\chi$:

    \begin{itemize}
        \item $\chi = \at{p} \in \At$. This follows immediately by expanding the hypotheses with \ref{cl:tor:and} and \ref{cl:tor:inf}, choosing the atom to be $\chi$.
        
        \item $\chi = \chi_1 \to \chi_2$. By \ref{cl:tor:to}, the conclusion is equivalent to $\sigma \suppNew{\baseB} \tau$. By \ref{cl:tor:inf}, this is equivalent to the following: for any $\baseC \baseGeq \baseB$, if $\suppNew{\baseC} \chi_1$, then $\suppNew{\baseC} \chi_2$. Therefore, fix an arbitrary $\baseC \baseGeq \baseB$ such that  $\suppNew{\baseC} \chi_1$.  By the induction hypothesis (IH), it suffices to show: (1) $\suppNew{\baseC} \phi \land \psi$ and (2) for any  $\baseD \baseGeq \baseC$, if $\suppNew{\baseD} \phi$ and $\suppNew{\baseD} \psi$,  then $\suppNew{\baseD} \chi_2$. By Lemma \ref{lem:Tor:monotonenew} on the first hypothesis we immediately get (1). For (2), fix an arbitrary base $\baseD \baseGeq \baseC$ such that $\suppNew{\baseD} \phi$, and $\suppNew{\baseD} \psi$. By the second hypothesis, we obtain $\suppNew{\baseD} \chi_1 \to \chi_2$ --- that is, $\chi_1 \suppNew{\baseD} \chi_2$. Hence, by \ref{cl:tor:inf} and \ref{lem:Tor:monotonenew} (since $\baseD \supseteq \baseB$) we have $\suppNew{\baseD} \chi_2$, as required.
        
        \item $\chi = \chi_1 \land \chi_2$. By \ref{cl:tor:andnew}, the conclusion is equivalent to the following: for any $\baseC \baseGeq \baseB$ and atomic $\at{p}$, if $\chi_1,\chi_2 \suppNew{\baseC} \at{p}$, then $\suppNew{\baseC} \at{p}$. Therefore, fix arbitrary $\baseC \baseGeq \baseB$ and $\at{p}$ such that $\chi_1,\chi_2 \suppNew{\baseC} \at{p}$. By \ref{cl:tor:inf}, for any $\baseD \baseGeq \baseC$, if $\suppNew{\baseD} \chi_1 \text{ and } \suppNew{\baseD} \chi_2$, then $\suppNew{\baseY} \at{p}$. We require to show $\suppNew{\baseC} \at{p}$. By the IH, it suffices to show the following: (1) $\suppNew{\baseC} \phi \land \psi$ and (2), for any $\baseE \baseGeq \baseC$, if $\suppNew{\baseE} \phi$  and $\suppNew{\baseE} \psi$, then $ \suppNew{\baseE} \at{p}$. Since $\baseB \baseLeq \baseC$, By Lemma \ref{lem:Tor:monotonenew} on the first hypothesis we immediately get (1). For (2), fix an arbitrary base $\baseE \baseGeq \baseC$ such that $\suppNew{\baseE} \phi$ and $\suppNew{\baseE} \psi$. By the second hypothesis, we obtain $\suppNew{\baseD} \at{p}$, as required.
        \item $\chi = \chi_1 \lor \chi_2$. By \ref{cl:tor:or}, the conclusion is equivalent to the following: for any $\baseC \baseGeq \baseB$ and atomic $\at{p}$, if $\chi_1 \suppNew{\baseC} \at{p}$ and $\chi_2 \suppNew{\baseC} \at{p}$, then $\suppNew{\baseC} \at{p}$. Therefore, fix an arbitrary base $\baseC \baseGeq \baseB$ and atomic $\at{p}$ such that $\chi_1 \suppNew{\baseC} \at{p}$ and $\chi_2 \suppNew{\baseC} \at{p}$. By the IH, it suffices to prove the following: (1) $\suppNew{\baseC} \phi \land \psi$ and $(2)$. for any $\baseD \baseGeq \baseC$, if $\suppNew{\baseD} \phi$ and $\suppNew{\baseD} \psi$, then $\suppNew{\baseD} \at{p}$. By Lemma \ref{lem:Tor:monotonenew} on the first hypothesis we immediately get (1). For (2), fix an arbitrary $\baseD \baseGeq \baseC$ such that $\suppNew{\baseD} \phi$ and $\suppNew{\baseD} \psi$. Since $\baseD \baseGeq \baseB$, we obtain $\suppNew{\baseD} \chi_1 \lor \chi_2$ by the second hypothesis. By \ref{cl:tor:or}, we obtain $\suppNew{\baseD} \at{p}$, as required.
        \item $\chi = \bot$. By \ref{cl:tor:bot}, the conclusion is equivalent to the following: $\suppNew{\baseB} \at{r}$ for all atomic $\at{r}$. By the IH, it suffices to prove the following: (1) $\suppNew{\baseB} \varphi \land \psi$ and  (2), for any $\baseC \baseGeq \baseB$, if $\suppNew{\baseC} \phi$ and $\suppNew{\baseC} \psi$, then $\suppNew{\baseC} \at{r}$. By the first hypothesis we have (1). For (2), fix an arbitrary $\baseC \baseGeq \baseB$ such that $\suppNew{\baseC} \phi$ and $\suppNew{\baseC} \psi$. By the second hypothesis, $\suppNew{\baseC} \bot$ obtains.  By \ref{cl:tor:bot}, we obtain $\suppNew{\baseC} \at{r}$, as required.
    \end{itemize}
    This completes the induction.
\end{proof}

\corSuppFormulaMeanInfAllAtoms*
\begin{proof}
    Let $\top$ be any formula such that $\suppTor{} \top$ --- for example, $\top := \at{p} \land (\at{p} \to \at{q}) \to \at{q}$. 

    We apply the two equivalent definitions of $\land$ to the neutrality of $\top$.  
    \begin{align*}
        \suppTor{\baseB} \varphi & \text{ iff } \suppTor{\baseB} \varphi \text{ and } \suppTor{\baseB} \top & \text{(def. of $\top$)} \\
        & \text{ iff } \suppTor{\baseB} \varphi \land \top & \text{\ref{cl:tor:and}} \\
        & \text{ iff for any } \baseX \baseGeq \baseB, \text{ for any } \at{p} \in \At, \varphi, \top \suppTor{\baseX} \at{p} \text{ implies } \suppTor{\baseB} \at{p} & \text{\ref{cl:tor:andnew}} \\
        & \text{ iff for any } \baseX \baseGeq \baseB, \text{ for any } \at{p} \in \At, \varphi \suppTor{\baseX} \at{p} \text{ implies } \suppTor{\baseB} \at{p} & \text{(def. of $\top$)}
    \end{align*}
    This establishes the desired equivalence. 
\end{proof}

\section{Omitted proofs from Section~\ref{subsec:MILL-atomic-system}} 
\label{app:MILL-support-in-base}

\begin{proposition}
    The supporting relation $\suppM{\baseB}{\at{S}}$ from Definition~\ref{def:MILL-BeS} is well-defined. 
\end{proposition}
\begin{proof}
    Basically we show that this is an inductive definition, by providing some metric. We follow the idea of Sandqvist, and notice that the extra layer of complexity given by the resource $\at{S}$ in $\suppM{\baseB}{\at{S}}$ does not impact the argument for well-definedness. 

    We define the \emph{degree} of \MILL{} formulas as follows: 
    \begin{align*}
        \degForm{ \at{p} } & \coloneqq 1 \\
        \degForm{ \mtop } & \coloneqq 2  \\
        \degForm{ \varphi \bullet \psi } & \coloneqq \degForm{ \varphi } + \degForm{ \psi } + 1, \text{ where } \bullet \in \{ \mand, \mto \}
    \end{align*}
    Note that for each of \eqref{eq:MILL-BeS-I}, \eqref{eq:MILL-BeS-tensor}, and \eqref{eq:MILL-BeS-imply}, the formulas appearing in the definitional clauses all have strictly smaller degrees than the formula itself, and the atomic case $\suppM{\baseB}{\at{S}}$ is defined by the derivability relation as $\at{S} \deriveBaseM{\baseB} \at{p}$. Therefore this is a valid inductive definition. 
\end{proof}

\MillBeSMonotonicity*
\begin{proof}
    Formally we prove by induction on $\suppM{}{}$ (see Definition~\ref{def:MILL-BeS}). 

    \begin{itemize}[label=--]
        \item For the base case, $\Gamma \suppM{\baseB}{\at{S}} \varphi$ is of the form $\suppM{\baseB}{\at{S}} \at{p}$ where $\at{p}$ is an atom. Then by definition this means $\at{S} \deriveBaseM{\baseB} \at{p}$. For arbitrary $\baseC$ that extends $\baseB$, $\at{S} \deriveBaseM{\baseC} \at{p}$ also holds simply because the derivability relation $\deriveBaseM{\baseX}$ is totally determined by the atomic rules in the base $\baseX$, and $\baseC \baseGeq \baseB$ means that every atomic ruls in $\baseB$ is also in $\baseC$. Then $\at{S} \deriveBaseM{\baseC} \at{p}$ says $\suppM{\baseC}{\at{S}} \at{p}$. 
        \item For the inductive cases \eqref{eq:MILL-BeS-tensor}, \eqref{eq:MILL-BeS-I}, \eqref{eq:MILL-BeS-imply} (expanded using \eqref{eq:MILL-BeS-Inf} and \eqref{eq:MILL-BeS-comma}), note that each uses a universal quantification over bases extending $\baseB$, namely `for every $\baseX \baseGeq \baseB$, ...'. Now for an arbitrary base $\baseC$ that extends $\baseB$, such universal quantified statement also holds by replacing the quantification with all bases extending $\baseC$, namely `for every $\baseX \baseGeq \baseC$, ...'. Therefore the inductive steps also pass. 
    \end{itemize}
    This completes the inductive proof. 
\end{proof}

\begin{corollary}
    $\Gamma \suppM{}{\at{S}} \varphi$ iff $\Gamma \suppM{\emptybase}{\at{S}} \varphi$. 
\end{corollary}
\begin{proof}
    Recall that the definition that $\Gamma \suppM{}{\at{S}} \varphi$ means the following: for any base $\baseB$, $\Gamma \suppM{\baseB}{\at{S}} \varphi$ holds.

    For the `only if' direction, note that $\Gamma \suppM{}{\at{S}} \varphi$ implies that in particular, $\Gamma \suppM{\emptybase}{\at{S}} \varphi$ holds. 

    For the `if' direction, suppose $\Gamma \suppM{\empty}{\at{S}} \varphi$ holds. Then for arbitrary $\baseB$, since $\baseB \baseGeq \emptybase$ holds, we can apply Proposition~\ref{lem:BeS-Mill-monotone-on-base} and conclude that $\Gamma \suppM{\baseB}{\at{S}} \varphi$ also holds. Since this is true for arbitrary base $\baseB$, we have $\Gamma \suppM{}{\at{S}} \varphi$. 
\end{proof}

\lemMILLAtomicCut*
\begin{proof}
It is straightforward to see that \eqref{eq:mill-atomic-cut-2} entails \eqref{eq:mill-atomic-cut-1}: we take $\baseX$ to be $\baseB$, and $\at{T}_i$ to be $\makeMultiset{\at{p}_i}$ for each $i = 1, \dots, n$. Since $\at{p}_1 \deriveBaseM{\baseB} \at{p}_1, \dots, \at{p}_n \deriveBaseM{\baseB} \at{p}_n$ all hold by \eqref{eq:derive-mill-ref}, it follows from \eqref{eq:mill-atomic-cut-2} that $\at{p}_1 \mcomma \dots \mcomma \at{p}_n \mcomma \at{S} \deriveBaseM{\baseB} \at{q}$, namely $\at{P} \mcomma \at{S} \deriveBaseM{\baseB} \at{q}$. 

    As for \eqref{eq:mill-atomic-cut-1} entails \eqref{eq:mill-atomic-cut-2}, we prove by induction on how $\at{P} \mcomma \at{S} \deriveBaseM{\baseB} \at{q}$ is derived (see Definition~\ref{def:derivability-base-MILL}). 
    \begin{itemize}[label={--}]
        \item $\at{P} \mcomma \at{S} \deriveBaseM{\baseB} \at{q}$ holds by \eqref{eq:derive-mill-ref}. That is, $\at{P} \mcomma \at{S} = \makeMultiset{\at{q}}$, and $\at{q} \deriveBaseM{\baseB} \at{q}$ follows by \eqref{eq:derive-mill-ref}. Here are two subcases, depending on which of $\at{P}$ and $\at{S}$ is $\makeMultiset{\at{q}}$. 
        \begin{itemize}[label={-}]
            \item Suppose $\at{P} = \makeMultiset{\at{q}}$ and $\at{S} = \emptymultiset$. So \eqref{eq:mill-atomic-cut-2} becomes: for every $\baseX \baseGeq \baseB$ and $\at{T}$, if $\at{T} \deriveBaseM{\baseX} \at{q}$, then $\at{T} \deriveBaseM{\baseX} \at{q}$. This holds a fortiori. 
            \item Suppose $\at{S} = \makeMultiset{\at{q}}$ and $\at{P} = \emptymultiset$. Since $\at{P} = \emptymultiset$, \eqref{eq:mill-atomic-cut-2} becomes: for every $\baseX \baseGeq \baseB$, $\at{S} \deriveBaseM{\baseX} \at{q}$. This holds by \eqref{eq:derive-mill-ref}. 
        \end{itemize}
        \item $\at{P}, \at{S} \deriveBaseM{\baseB} \at{q}$ holds by \eqref{eq:derive-mill-app}. We assume that $\at{P} = \at{P}_1 \mcomma \dots \mcomma \at{P}_k$, $\at{S} = \at{S}_1 \mcomma \dots \mcomma \at{S}_k$, and the following hold for some $\at{Q}_1, \dots, \at{Q}_k$ and $\at{r}_1, \dots, \at{r}_k$: 
        \begin{align*}
            & \at{P}_1 \mcomma \at{S}_1 \mcomma \at{Q}_1 \deriveBaseM{\baseB} \at{r}_1, \dots, \at{P}_k \mcomma \at{S}_k \mcomma \at{Q}_k \deriveBaseM{\baseB} \at{r}_k \hfill \tag{\ref{eq:mill-atomic-cut-proof-1}} \\
            & \left( \at{Q}_1 \seq \at{r}_1, \dots, \at{Q}_k \seq \at{r}_k \right) \Rightarrow \at{q} \text{ is in } \baseB \tag{\ref{eq:mill-atomic-cut-proof-2}}
        \end{align*}
        \labelandtag{eq:mill-atomic-cut-proof-1}{3}
        \labelandtag{eq:mill-atomic-cut-proof-2}{4}
        In order to prove \eqref{eq:mill-atomic-cut-2}, we fix some arbitrary base $\baseC \baseGeq \baseB$ and atomic multisets $\at{T}_1, \dots, \at{T}_n$ such that $\at{T}_1 \deriveBaseM{\baseC} \at{p}_1, \dots, \at{T}_n \deriveBaseM{\baseC} \at{p}_n$, and show $\at{T}_1 \mcomma \dots \mcomma \at{T}_n \mcomma \at{S} \deriveBaseM{\baseC} \at{q}$.  Let us assume $\at{P}_i = \at{p}_{i1} \mcomma \dots \mcomma \at{p}_{i \ell_{i}}$ for each $i = 1, \dots, k$. We apply IH to every $\at{P}_i \mcomma \at{S}_i \mcomma \at{Q}_i \deriveBaseM{\baseB} \at{r}_i$ from \eqref{eq:mill-atomic-cut-proof-1}, and get $\at{T}_{i1} \mcomma \dots \mcomma \at{T}_{i\ell_{i}} \mcomma \at{S}_i \mcomma \at{Q}_i \deriveBaseM{\baseC} \at{r}_i$. Moreover, the atomic rule from \eqref{eq:mill-atomic-cut-proof-2} is also in $\baseC$, since $\baseC \baseGeq \baseB$. Therefore we can apply \eqref{eq:derive-mill-app} and get 
        \[\at{T}_{11} \mcomma \dots \mcomma \at{T}_{1\ell_{1}} \mcomma \at{S}_1 \mcomma \dots \mcomma \at{T}_{k1} \mcomma \dots \mcomma \at{T}_{k\ell_{k}} \mcomma \at{S}_k \deriveBaseM{\baseC} \at{q}. 
         \] 
        By the definition of $\at{S}_i$ and $\at{T}_{i j}$, this is precisely $\at{T}_1 \mcomma \dots \mcomma \at{T}_n \mcomma \at{S} \deriveBaseM{\baseC} \at{q}$. 
    \end{itemize}
    This completes the inductive proof. 
\end{proof}

\section{Proof of Soundness}
\label{app:MILL-BeS-soundness}
This appendix is devoted to the detailed proof of the soundness (Theorem~\ref{thm:MILL-BeS-soundness}) of the base-extension semantics for \MILL{}. 
\thmMILLBeSSoundness*
\begin{proof}
    Recall that $\Gamma \suppM{}{} \phi$ is abbreviation of $\Gamma \suppM{\emptybase}{\emptymultiset} \phi$. By the inductive definition of $\provesMILL$, it suffices to prove the following: 
    \begin{description}
        \item[Ax\label{eq:mill-soundness-axiom}] $\phi \suppM{}{} \phi$ 
        \item[$\mto$I\label{eq:mill-soundness-implication-intro}] If $\Gamma \mcomma \phi \suppM{}{} \psi$, then $\Gamma \suppM{}{} \phi \mto \psi$. 
        \item[$\mto$E\label{eq:mill-soundness-implication-elim}] If $\Gamma \suppM{}{} \phi \mto \psi$ and $\Delta \suppM{}{} \phi$, then $\Gamma \mcomma \Delta \suppM{}{} \psi$. 
        \item[$\mand$I\label{eq:mill-soundness-conjunction-intro}] If $\Gamma \suppM{}{} \phi$ and $\Delta \suppM{}{} \psi$, then $\Gamma \mcomma \Delta \suppM{}{} \phi \mand \psi$. 
        \item[$\mand$E\label{eq:mill-soundness-conjunction-elim}] If $\Gamma \suppM{}{} \phi \mand \psi$ and $\Delta \mcomma \phi \mcomma \psi \suppM{}{} \chi$, then $\Gamma \mcomma \Delta \suppM{}{} \chi$. 
        \item[$\mtop$I\label{eq:mill-soundness-top-intro}] $\suppM{}{} \mtop$
        \item[$\mtop$E\label{eq:mill-soundness-top-elim}] If $\Gamma \suppM{}{} \chi$ and $\Delta \suppM{}{} \mtop$, then $\Gamma \mcomma \Delta \suppM{}{} \chi$. 
    \end{description}
    Now we prove them one by one. We assume that $\Gamma = \gamma_1 \mcomma \dots \mcomma \gamma_m$ and $\Delta = \delta_1 \mcomma \dots \mcomma \delta_n$ in all the above equations to be checked. 
    \begin{itemize}[label={-}]
        \item \eqref{eq:mill-soundness-axiom} holds a fortiori by definition of the validity relation $\suppM{}{}$: by \eqref{eq:MILL-BeS-Inf}, $\phi \suppM{}{} \phi$ means that for every base $\baseX$, if $\suppM{\baseX}{} \phi$, then $\suppM{\baseX}{}$. 

        \item \eqref{eq:mill-soundness-implication-intro}. Assume $\Gamma, \phi \suppM{}{} \psi$, we show $\Gamma \suppM{}{} \phi \mto \psi$. By \eqref{eq:MILL-BeS-Inf}, the assumption $\Gamma \mcomma \phi \suppM{}{} \psi$ boils down to the following: 
        \begin{equation}\label{eq:mill-soundness-implication-intro-proof-1}
        \begin{split}
            \text{For all base } \baseX & \text{ and multiset of atoms } \at{P}, \text{ if there exists } \at{S}_1, \dots, \at{S}_m, \at{T} \\
            & \text{ satisfying } \at{P} = \at{S}_1 \mcomma \dots \mcomma \at{S}_m \mcomma \at{T},  \\
            & \text{ such that } \suppM{\baseX}{\at{S}_1} \gamma_1, \dots, \suppM{\baseX}{\at{S}_m} \gamma_m, \suppM{\baseX}{\at{T}} \phi, \text{ then } \suppM{\baseX}{\at{P}} \psi. 
        \end{split}
        \end{equation}
        In order to show that $\Gamma \suppM{\baseB}{\at{P}} \psi$, we fix an arbitrary base $\baseB$ and multiset of atoms $\at{P}$ satisfying that there exists $\at{P}_1, \dots, \at{P}_m$ such that $\at{P} = \at{P}_1 \mcomma \dots \mcomma \at{P}_m$, and $\suppM{\baseB}{\at{P}_1} \gamma_1, \dots, \suppM{\baseB}{\at{P}_m} \gamma_m$. The goal is show that $\suppM{\baseB}{\at{P}} \phi \mto \psi$. 
        By \eqref{eq:MILL-BeS-imply}, $\suppM{\baseB}{\at{P}} \phi \mto \psi$ means $\phi \suppM{\baseB}{\at{P}} \psi$. To show $\phi \suppM{\baseB}{\at{P}} \psi$, we fix an arbitrary $\baseC \baseGeq \baseB$ and multiset $\at{Q}$ such that $\suppM{\baseC}{\at{Q}} \phi$, and prove that $\suppM{\baseC}{\at{P}, \at{Q}} \psi$. By monotonicity of $\suppM{}{}$ with respect to the base, $\suppM{\baseB}{\at{P}_i} \gamma_i$ implies $\suppM{\baseC}{\at{P}_i} \gamma_i$, for $i = 1, \dots, m$. Apply \eqref{eq:mill-soundness-implication-intro-proof-1} to this together with $\suppM{\baseC}{\at{Q}} \phi$, it follows that $\suppM{\baseC}{\at{P}, \at{Q}} \psi$. 

        \item \eqref{eq:mill-soundness-implication-elim}. Assume $\Gamma \suppM{}{} \phi \mto \psi$ and $\Delta \suppM{}{} \phi$, we show $\Gamma \mcomma \Delta \suppM{}{} \psi$. Spelling out the definition of $\Gamma \suppM{}{} \phi \mto \psi$ and $\Delta \suppM{}{} \phi$ using \eqref{eq:MILL-BeS-Inf}, we have: 
        \begin{equation}
        \label{eq:mill-soundness-implication-elim-proof-1}
            \begin{split}
                \text{For every base } \baseX & \text{ and atomic multisets } \at{P} = \at{P}_1 \mcomma \dots \mcomma \at{P}_m, \\
                & \text{ if } \suppM{\baseX}{\at{P}_1} \gamma_1, \dots, \suppM{\baseX}{\at{P}_m} \gamma_m, \text{ then } \suppM{\baseX}{\at{P}} \phi \mto \psi.      
            \end{split}
        \end{equation}
        \begin{equation}
        \label{eq:mill-soundness-implication-elim-proof-2}
            \begin{split}
                \text{For every base } \baseY & \text{ and atomic multisets } \at{Q} = \at{Q}_1 \mcomma \dots \mcomma \at{Q}_{n}, \\
                & \text{ if } \suppM{\baseY}{\at{Q}_1} \delta_1, \dots, \suppM{\baseY}{\at{Q}_n} \delta_n, \text{ then } \suppM{\baseY}{\at{Q}} \phi. 
            \end{split}
        \end{equation}
        In order to show $\Gamma \mcomma \Delta \suppM{}{} \psi$, we fix an arbitrary base $\baseB$, atomic multisets $\at{S} = \at{S}_1 \mcomma \dots \mcomma \at{S}_m$ and $\at{T} = \at{T}_1 \mcomma \dots \mcomma \at{T}_n$, such that $\suppM{\baseB}{\at{S}_1} \gamma_1, \dots, \suppM{\baseB}{\at{S}_m} \gamma_m$ and $\suppM{\baseB}{\at{T}_1} \delta_1, \dots, \suppM{\baseB}{\at{T}_n} \delta_n$, and go on to prove that $\suppM{\baseB}{\at{S} \mmcomma \at{T}} \psi$. Using \eqref{eq:mill-soundness-implication-elim-proof-1}, $\suppM{\baseB}{\at{S}_1} \gamma_1, \dots, \suppM{\baseB}{\at{S}_m} \gamma_m$ implies that $\suppM{\baseB}{\at{S}} \phi \mto \psi$; using \eqref{eq:mill-soundness-implication-elim-proof-2}, $\suppM{\baseB}{\at{T}_1} \delta_1, \dots, \suppM{\baseB}{\at{T}_n} \delta_n$ implies that $\suppM{\baseB}{\at{T}} \phi$. Spelling out the definition of $\suppM{\baseB}{\at{S}} \phi \mto \psi$, we know that for arbitrary base $\baseX \baseGeq \baseB$ and atomic multiste $\at{U}$, if $\suppM{\baseX}{\at{U}} \phi$, then $\suppM{\baseX}{\at{S} \mcomma \at{U}} \psi$. In particular, since $\suppM{\baseB}{\at{T}} \phi$, we have $\suppM{\baseB}{\at{S} \mmcomma \at{T}} \psi$. 
        \item \eqref{eq:mill-soundness-conjunction-intro}. We assume $\Gamma \suppM{}{} \phi$ and $\Delta \suppM{}{} \psi$, and show that $\Gamma \mcomma \Delta \suppM{}{} \phi \mand \psi$ holds. Spelling out the definition of $\Gamma \mcomma \Delta \suppM{}{} \phi \mand \psi$, it suffices to fix some base $\baseB$ and atomic multiset $\at{S}_1, \dots, \at{S}_m, \at{T}_1, \dots, \at{T}_n$ (denote $\at{S} = \at{S}_1 \mcomma \dots \mcomma \at{S}_m$, and $\at{T} = \at{T}_1 \mcomma \dots \mcomma \at{T}_n$) such that $\suppM{\baseB}{\at{S}_1} \gamma_1, \dots, \suppM{\baseB}{\at{S}_m} \gamma_m, \suppM{\baseB}{\at{T}_1} \delta_1, \dots, \suppM{\baseB}{\at{T}_n} \delta_n$, and show that $\suppM{\baseB}{\at{S}, \at{T}} \phi \mand \psi$ follows. According to \eqref{eq:MILL-BeS-tensor}, we can simply fix some base $\baseC \baseGeq \baseB$, atomic multiset $\at{U}$, atom $\at{p}$ satisfying $\phi \mcomma \psi \suppM{\baseC}{\at{U}} \at{p}$, and show $\suppM{\baseC}{\at{S} \mmcomma \at{T} \mmcomma \at{U}} \at{p}$. 
        Note that $\suppM{\baseB}{\at{S} \mmcomma \at{T}} \makeMultiset{ \phi, \psi }$ holds: $\suppM{\baseB}{\at{S}} \Gamma$ and $\Gamma \suppM{}{} \phi$ implies $\suppM{\baseB}{\at{S}} \phi$; $\suppM{\baseB}{\at{T}} \Delta$ implies $\suppM{\baseB}{\at{S}} \psi$. By monotonicity, $\suppM{\baseC}{\at{S} \mmcomma \at{T}} \makeMultiset{ \phi, \psi }$. This together with $\phi \mcomma \psi \suppM{\baseC}{\at{U}} \at{p}$ entail that $\suppM{\baseC}{\at{S} \mmcomma \at{T} \mmcomma \at{U}} \at{p}$. 
        \item \eqref{eq:mill-soundness-conjunction-elim}. We use Lemma~\ref{lem:mill-conj-key-lemma}. Suppose $\Gamma \suppM{}{} \phi \mand \psi$ and $\Delta \mcomma \phi \mcomma \psi \suppM{}{} \chi$, and we show that $\Gamma \mcomma \Delta \suppM{}{} \chi$. So let us suppose that $\suppM{\baseB}{\at{S}} \Gamma$ and $\suppM{\baseB}{\at{T}} \Delta$ (thus $\suppM{\baseB}{\at{S} \mmcomma \at{T}} \Gamma \mcomma \Delta$), and show that $\suppM{\baseB}{\at{S} \mmcomma \at{T}} \chi$. First, $\Gamma \suppM{}{} \phi \mand \psi$ together with $\suppM{\baseB}{\at{S}} \Gamma$ entails that $\suppM{\baseB}{\at{S}} \phi \mand \psi$. Second, spelling out the definition of $\Delta \mcomma \phi \mcomma \psi \suppM{}{} \chi$, we have: 
        \begin{equation}
        \label{eq:mill-soundness-conj-elim-proof-2}
            \begin{split}
                \text{For every base } \baseX & \text{ and atomic multisets } \at{P}, \at{Q}, \\
                & \text{ if } \suppM{\baseX}{\at{P}} \Delta \text{ and } \suppM{\baseX}{\at{Q}} \makeMultiset{ \phi, \psi }, \text{ then } \suppM{\baseX}{\at{P} \mmcomma \at{Q}} \chi. 
            \end{split}
        \end{equation}
        Under the assumption $\suppM{\baseB}{\at{T}} \Delta$, by fixing $\at{P}$ and $\at{Q}$ to be $\at{T}$ and $\at{S}$ respectively, \eqref{eq:mill-soundness-conj-elim-proof-2} implies the following: 
        \begin{equation}
        \label{eq:mill-soundness-conj-elim-proof-1}
        \text{For every base } \baseX \baseGeq \baseB, \text{ if } \suppM{\baseX}{\at{S}} \makeMultiset{ \phi, \psi }, \text{ then } \suppM{\baseX}{\at{T} \mmcomma \at{S}} \chi.  
        \end{equation}
        Now, given $\suppM{\baseB}{\at{S}} \phi \mand \psi$ and \eqref{eq:mill-soundness-conj-elim-proof-1}, we can apply Lemma~\ref{lem:mill-conj-key-lemma} and conclude that $\suppM{\baseB}{\at{S} \mmcomma \at{T}} \chi$. 
        \item \eqref{eq:mill-soundness-top-intro}. By \eqref{eq:MILL-BeS-I}, $\suppM{}{} \mtop$ is equivalent to that for every base $\baseX$, atomic multiset $\at{U}$, and atom $\at{q}$, if $\suppM{\baseX}{\at{U}} \at{q}$, then $\suppM{\baseX}{\at{U}} \at{q}$. This is true a fortiori. 
        \item \eqref{eq:mill-soundness-top-elim}. We assume $\Gamma \suppM{}{} \chi$ and $\Delta \suppM{}{} \mtop$, and show that $\Gamma \mcomma \Delta \suppM{}{} \chi$. Towards this, we fix some base $\baseB$ and atomic multisets $\at{S}, \at{T}$ such that $\suppM{\baseB}{\at{S}} \Gamma$ and $\suppM{\baseB}{\at{T}} \Delta$, and show that $\suppM{\baseB}{\at{S} \mmcomma \at{T}} \chi$. 
        By $\Delta \suppM{}{} \mtop$ and $\suppM{\baseB}{\at{T}} \Delta$, we know that $\suppM{\baseB}{\at{T}} \mtop$. 
        By $\Gamma \suppM{}{} \chi$ and $\suppM{\baseB}{ \at{S}} \Gamma$, we have $\suppM{\baseB}{\at{S}} \chi$. 
        Now apply Lemma~\ref{lem:mill-mtop-key-lemma} to $\suppM{\baseB}{\at{T}} \mtop$ and $\suppM{\baseB}{\at{S}} \chi$, we conclude that $\suppM{\baseB}{\at{S} \mmcomma \at{T}} \chi$. 
    \end{itemize}
    This completes the verification of all items. 
\end{proof}

\lemMILLConjKeyLemma*
\begin{proof}
    We prove by induction on the structure of $\chi$. 
    The condition \eqref{eq:mill-soundness-key-lemma-condition-2} can be spelled out as: for every $\baseX \baseGeq \baseB$ and $\at{U}$, if $\suppM{\baseX}{\at{U}} \makeMultiset{ \phi, \psi }$, then $\suppM{\baseX}{\at{U} \mmcomma \at{T}} \chi$. 
    \begin{itemize}[label={-}]
        \item When $\chi$ is an atom, the statement of the lemma follows immediately from \eqref{eq:MILL-BeS-tensor}. 
        \item $\chi = \mtop$. By \eqref{eq:MILL-BeS-I}, \eqref{eq:mill-soundness-key-lemma-conclusion} amounts to that, for every $\baseX \baseGeq \baseB$, atomic multiset $\at{U}$, atom $\at{p}$, if $\suppM{\baseX}{\at{U}} \at{p}$, then $\suppM{\baseX}{\at{S} \mmcomma \at{T} \mmcomma \at{U}} \at{p}$. So we fix some base $\baseC \baseGeq \baseB$, atomic multiset $\at{Q}$, and atom $\at{q}$, such that $\suppM{\baseC}{\at{Q}} \at{q}$. The goal is to show $\suppM{\baseC}{\at{S} \mmcomma \at{T} \mmcomma \at{Q}} \at{q}$. 
        According to the atomic case, this follows from the following two facts: 
        \begin{align}
            & \suppM{\baseC}{\at{S}} \phi \mand \psi \label{eq:mill-soundness-key-lemma-proof-7} \\
            & \phi \mcomma \psi \suppM{\baseC}{\at{T} \mmcomma \at{Q}} \at{q} \label{eq:mill-soundness-key-lemma-proof-8}
        \end{align}
        Here \eqref{eq:mill-soundness-key-lemma-proof-7} follows immediately from \eqref{eq:mill-soundness-key-lemma-condition-1} and $\baseC \baseGeq \baseB$, so it suffices to prove \eqref{eq:mill-soundness-key-lemma-proof-8}. For this, we fix some base $\baseD \baseGeq \baseC$, atomic multiset $\at{R}_1, \at{R}_2$ such that $\suppM{\baseD}{\at{R}_1} \phi$ and $\suppM{\baseD}{\at{R}_2} \psi$ hold, and show that $\suppM{\baseD}{\at{T} \mmcomma \at{Q} \mmcomma \at{R}_1 \mmcomma \at{R}_2} \at{q}$. Note that \eqref{eq:mill-soundness-key-lemma-condition-2} now becomes $\phi \mcomma \psi \suppM{\baseB}{\at{T}} \mtop$. So together with $\suppM{\baseD}{\at{R}_1} \phi$ and $\suppM{\baseD}{\at{R}_2} \psi$, it follows that $\suppM{\baseD}{\at{T} \mmcomma \at{R}_1 \mmcomma \at{R}_2} \mtop$. This according to \eqref{eq:MILL-BeS-I} says that for every $\baseX \baseGeq \baseD$, atomic multiset $\at{U}$, and atom $\at{p}$, $\suppM{\baseX}{\at{U}} \at{p}$ implies $\suppM{\baseX}{\at{T} \mmcomma \at{R}_1 \mmcomma \at{R}_2 \mmcomma \at{U}} \at{p}$. In particular, since $\suppM{\baseD}{\at{Q}} \at{q}$ (which is immediately consequence of $\suppM{\baseC}{\at{Q}} \at{q}$ and $\baseD \baseGeq \baseC$), it follows that $\suppM{\baseD}{\at{T} \mmcomma \at{R}_1 \mmcomma \at{R}_2 \mmcomma \at{Q}} \at{q}$. 
        \item $\chi = \sigma \mto \tau$. The goal is to prove that, given \eqref{eq:mill-soundness-key-lemma-condition-1} and \eqref{eq:mill-soundness-key-lemma-condition-2}, $\suppM{\baseB}{\at{S} \mmcomma \at{T}} \sigma \mto \tau$ holds; spelling out the definition using \eqref{eq:MILL-BeS-imply} and \eqref{eq:MILL-BeS-Inf}, this amounts to showing that for arbitrary $\baseX \baseGeq \baseB$ and atomic multiset $\at{U}$, if $\suppM{\baseX}{\at{U}} \sigma$, then $\suppM{\baseX}{\at{S} \mmcomma \at{T} \mmcomma \at{U}} \tau$. 
        So we fix an arbitrary $\baseC \baseGeq \baseB$ and atomic multiset $\at{P}$ such that $\suppM{\baseC}{\at{P}} \sigma$ holds, and the goal is to show $\suppM{\baseC}{\at{S} \mmcomma \at{T} \mmcomma \at{P}} \tau$. By IH, it suffices to show the following: 
        \begin{align}
            & \suppM{\baseC}{\at{S}} \phi \mand \psi 
            \label{eq:mill-soundness-key-lemma-proof-1} \\
            & \phi \mcomma \psi \suppM{\baseC}{\at{T} \mmcomma \at{P}} \tau \label{eq:mill-soundness-key-lemma-proof-2} 
        \end{align}
        Since \eqref{eq:mill-soundness-key-lemma-proof-1} is exactly \eqref{eq:mill-soundness-key-lemma-condition-1}, we focus on \eqref{eq:mill-soundness-key-lemma-proof-2}. So we fix an arbitrary $\baseD \baseGeq \baseC$ and $\at{Q}$ such that $\suppM{\baseD}{\at{Q}} \makeMultiset{\phi, \psi}$, and show $\suppM{\baseD}{\at{Q} \mmcomma \at{T} \mmcomma \at{P}} \tau$. Apply \eqref{eq:mill-soundness-key-lemma-condition-2} to $\suppM{\baseD}{\at{Q}} \makeMultiset{\phi, \psi}$, we get $\suppM{\baseD}{\at{Q} \mmcomma \at{T}} \sigma \mto \tau$, or equivalently $\sigma \suppM{\baseD}{\at{Q} \mmcomma \at{T}} \tau$. That is, for every $\baseY \baseGeq \baseD$ and atomic multiset $\at{U}$, $\suppM{\baseY}{\at{U}} \sigma$ implies $\suppM{\baseY}{\at{Q} \mmcomma \at{T} \mmcomma \at{U}} \tau$. Therefore, given $\suppM{\baseC}{\at{P}} \sigma$, by monotonicity we have $\suppM{\baseD}{\at{P}} \sigma$, thus $\suppM{\baseD}{\at{Q} \mmcomma \at{T} \mmcomma \at{P}} \tau$. 
        \item $\chi = \sigma \mand \tau$. Given \eqref{eq:mill-soundness-key-lemma-condition-1} and \eqref{eq:mill-soundness-key-lemma-condition-2}, we show $\suppM{\baseB}{\at{S} \mmcomma \at{T}} \sigma \mand \tau$. Spelling out the definition using \eqref{eq:MILL-BeS-tensor}, we can simply fix an arbitrary $\baseC \baseGeq \baseB$, atomic multiset $\at{P}$, and atom $\at{p}$ such that $\sigma \mcomma \tau \suppM{\baseC}{\at{P}} \at{p}$; in other words, 
        \begin{equation}
        \label{eq:mill-soundness-key-lemma-proof-5}
            \text{ for every } \baseX \baseGeq \baseC \text{ and } \at{U}, \text{ if } \suppM{\baseX}{\at{U}} \sigma \mcomma \tau, \text{ then } \suppM{\baseX}{\at{U} \mmcomma \at{P}} \at{p} 
        \end{equation}
        and then show $\suppM{\baseC}{\at{S} \mmcomma \at{T} \mmcomma \at{P}} \at{p}$. By IH, it suffices to prove the following: 
        \begin{align}
            & \suppM{\baseC}{\at{S}} \phi \mand \psi 
            \label{eq:mill-soundness-key-lemma-proof-3} \\
            & \phi \mcomma \psi \suppM{\baseC}{\at{T} \mmcomma \at{P}} \at{p} \label{eq:mill-soundness-key-lemma-proof-4} 
        \end{align}
        Now \eqref{eq:mill-soundness-key-lemma-proof-3} follows immediately from \eqref{eq:mill-soundness-key-lemma-condition-1} by monotonicity. Towards \eqref{eq:mill-soundness-key-lemma-proof-4}, let us fix arbitrary $\baseD \baseGeq \baseC$ and $\at{Q}$ such that $\suppM{\baseD}{\at{Q}} \phi \mcomma \psi$, and prove $\suppM{\baseD}{\at{Q} \mmcomma \at{T} \mmcomma \at{P}} \at{p}$. By \eqref{eq:mill-soundness-key-lemma-condition-2}, $\suppM{\baseD}{\at{Q}} \phi \mcomma \psi$ entails that $\suppM{\baseD}{\at{Q} \mmcomma \at{T}} \sigma \mand \tau$. This by \eqref{eq:MILL-BeS-tensor} means that, 
        \begin{equation}
        \label{eq:mill-soundness-key-lemma-proof-6}
        \text{for every } \baseY \baseGeq \baseD, \at{V} \text{ and } \at{q}, \text{ if } \sigma \mcomma \tau \suppM{\baseY}{\at{V}} \at{q}, \text{ then } \suppM{\baseY}{\at{Q} \mmcomma \at{T} \mmcomma \at{V}} \at{q}. 
        \end{equation}
        In particular, since $\sigma \mcomma \tau \suppM{\baseD}{\at{P}} \at{p}$, we can conclude from \eqref{eq:mill-soundness-key-lemma-proof-6} that $\suppM{\baseD}{\at{Q} \mmcomma \at{T} \mmcomma \at{P}} \at{p}$. 
    \end{itemize}
    This completes all the cases of the proof by induction. 
\end{proof}

\lemMILLMtopKeyLemma*
\begin{proof}
    We prove by induction on the structure of $\chi$. 

    \begin{itemize}[label={-}]
        \item $\chi$ is some atom $\at{q}$. Spelling out the definition of \eqref{eq:mill-soundness-mtop-key-lemma-condition-1} $\suppM{\baseB}{\at{S}} \mtop$, we have that for arbitrary $\baseX \baseGeq \baseB$, atomic multiset $\at{U}$, and atom $\at{p}$, if $\suppM{\baseX}{\at{U}} \at{p}$, then $\suppM{\baseX}{\at{S} \mmcomma \at{U}} \at{p}$. Apply this to \eqref{eq:mill-soundness-mtop-key-lemma-condition-2} $\suppM{\baseB}{\at{T}} \at{q}$, it follows that $\suppM{\baseB}{\at{S} \mmcomma \at{T}} \at{q}$. 
        \item $\chi = \mtop$. In order to prove $\suppM{\baseB}{\at{S} \mmcomma \at{T}} \mtop$, it suffices to fix some base $\baseC \baseGeq \baseB$, atomic multiset $\at{W}$, and atom $\at{q}$ such that $\suppM{\baseC}{\at{W}} \at{q}$, and prove that $\suppM{\baseC}{\at{S} \mmcomma \at{T} \mmcomma \at{W}} \at{q}$. Since $\suppM{\baseB}{\at{S}} \mtop$, $\baseC \baseGeq \baseB$, and $\suppM{\baseC}{\at{W}} \at{q}$, we have $\suppM{\baseC}{\at{S} \mmcomma \at{W}} \at{q}$. This together with $\suppM{\baseB}{\at{T}} \at{q}$ and $\baseC \baseGeq \baseB$ imply that $\suppM{\baseC}{\at{S} \mmcomma \at{T} \mmcomma \at{W}} \at{q}$. 
        \item $\chi = \sigma \mand \tau$. Uses Lemma~\ref{lem:mill-conj-key-lemma}. The goal is to show that $\suppM{\baseB}{\at{S} \mmcomma \at{T}} \sigma \mand \tau$; using \eqref{eq:MILL-BeS-tensor}, for every $\baseX \baseGeq \baseB$, $\at{U}$, $\at{p}$, if $\sigma \mcomma \tau \suppM{\baseX}{\at{U}} \at{p}$, then $\suppM{\baseX}{\at{S} \mmcomma \at{T} \mmcomma \at{U}} \at{p}$. So we fix some base $\baseC \baseGeq \baseB$, atomic multiset $\at{W}$, and atom $\at{q}$ such that $\sigma \mcomma \tau \suppM{\baseC}{\at{W}} \at{q}$, and the goal is now to show that $\suppM{\baseC}{\at{S} \mmcomma \at{T} \mmcomma \at{W}} \at{q}$. Apply Lemma~\ref{lem:mill-conj-key-lemma} to $\suppM{\baseC}{\at{T}} \sigma \mand \tau$ (which follows immediately from $\suppM{\baseB}{\at{T}} \sigma \mand \tau$ and $\baseC \baseGeq \baseB$) and $\sigma \mcomma \tau \suppM{\baseC}{\at{W}} \at{q}$, we have $\suppM{\baseC}{\at{T} \mmcomma \at{W}} \at{q}$. Together with $\suppM{\baseB}{\at{S}} \mtop$, we can conclude that $\suppM{\baseC}{\at{S} \mmcomma \at{T} \mmcomma \at{W}} \at{q}$. 
        \item $\chi = \sigma \mto \tau$. Spelling out the definition \eqref{eq:MILL-BeS-imply}, the goal $\suppM{\baseB}{\at{S} \mmcomma \at{T}} \sigma \mto \tau$ is equivalent to $\sigma \suppM{\baseB}{\at{S} \mmcomma \at{T}} \tau$. So we fix some base $\baseC \baseGeq \baseB$ and atomic multiset $\at{W}$ such that $\suppM{\baseC}{\at{W}} \sigma$, and then show that $\suppM{\baseC}{\at{S} \mmcomma \at{T} \mmcomma \at{W}} \tau$. By IH, from $\suppM{\baseC}{\at{S}} \mtop$ and $\suppM{\baseC}{\at{W}} \sigma$, we have $\suppM{\baseC}{\at{S} \mmcomma \at{W}} \sigma$. This together with $\suppM{\baseB}{\at{T}} \sigma \mto \tau$ implies that $\suppM{\baseC}{\at{S} \mmcomma \at{T} \mmcomma \at{W}} \tau$. 
    \end{itemize}
    This completes all the inductive cases. 
\end{proof}

\section{Proof of Completeness}
\label{app:MILL-BeS-completeness}

\begin{proposition}[\ref{eq:mill-atomic-sound-and-complete}]
\label{prop:mill-atomic-sound-and-complete}
    For arbitrary base $\baseB$, atomic multisets $\at{P}, \at{S}$, and atom $\at{q}$, 
    \[ 
        \at{P} \mcomma \at{S} \deriveBaseM{\baseB} \at{q} \text{ iff } \at{P} \suppM{\baseB}{\at{S}} \at{q}.
    \]
\end{proposition}
\begin{proof}
    The equivalence follows immediately from Proposition~\ref{lem:mill-atomic-cut}. Let us assume that $\at{P} = \makeMultiset{ \at{p}_1, \dots, \at{p}_n }$. Starting from $\at{P} \suppM{\baseB}{\at{S}} \at{q}$, by \eqref{eq:MILL-BeS-Inf}, it means for every base $\baseX \baseGeq \baseB$ and atomic multisets $\at{T}_1, \dots, \at{T}_n$, $\suppM{\baseX}{\at{T}_1} \at{p}_1, \dots, \suppM{\baseX}{\at{T}_n} \at{p}_n$ implies $\suppM{\baseX}{\at{S} \mmcomma \at{T}} \at{q}$. Spelling out the definition of $\suppM{\baseB}{}$ for atoms \eqref{eq:MILL-BeS-At}, $\at{P} \suppM{\baseB}{\at{S}} \at{q}$ is equivalent to that, for every base $\baseX \baseGeq \baseB$ and atomic multisets $\at{T}_1, \dots, \at{T}_n$, $\at{T}_1 \deriveBaseM{\baseX} \at{p}_1, \dots, \at{T}_n \deriveBaseM{\baseX} \at{p}_n$ implies $\at{S} \mcomma \at{T} \deriveBaseM{\baseX} \at{q}$. This is precisely $\at{P} \mcomma \at{S} \deriveBaseM{\baseB} \at{q}$, given Proposition~\ref{lem:mill-atomic-cut}. 
\end{proof}

\thmMILLBeSCompleteness*
\begin{proof}
    We assume $\Gamma \suppM{}{} \phi$, and $\Gamma = \makeMultiset{\gamma_1, \dots, \gamma_n}$. 
    Let $\Xi$ be $\subform{ \Gamma \cup \{ \phi \} }$, namely the set of all subformulas  $\Gamma$ and $\phi$. Since $\Gamma \cup \{ \phi \}$ is finite, $\Xi$ is also a finite set. 
    We define a `flattening' function $\flatmill{(\cdot)} \colon \Xi \to \At$: 
    it assigns to each non-atomic $\xi \in \Xi$ a unique atom which does not appear in $\Xi$, denoted as $\flatmill{\xi}$ (uniqueness means $\flatmill{\xi} \neq \flatmill{\zeta}$ if $\xi \neq \zeta$); 
    for each atomic $\at{p} \in \Xi$, we define $\flatmill{\at{p}}$ to be $\at{p}$ itself. 
    Conversely, we define the `deflattening' function $\deflatmill{(\cdot)} \colon \At \to \Xi \cup \At$ as an extension of the inverse of $\flatmill{(\cdot)}$: 
    for every atom in the image of $\flatmill{(\cdot)}$ say $\flatmill{\gamma}$ (note that such $\gamma$ is unique if it exists), we define $\deflatmill{ (\flatmill{\gamma}) }$ as $\gamma$; 
    for the other atoms, $\deflatmill{(\cdot)}$ is simply identity.
    We generalize both notations to multisets of formulas: $\flatmill{\Delta} \coloneqq \makeMultiset{\flatmill{\delta} \mid \delta \in \Delta}$ and $\flatmill{\at{P}} \coloneqq \makeMultiset{\flatmill{\at{p}} \mid \at{p} \in \at{P}}$; likewise for $\deflatmill{(\cdot)}$. 
    
    We still construct the base $\baseMill$ that encodes the natural deduction for \MILL{}. Base $\baseMill$ contains the following atomic rules, where $\sigma$ and $\tau$ range over $\Gamma \cup \{ \phi \}$, and $\at{p}$ ranges over all atoms: 
    \begin{enumerate}[label=(\arabic*)]
        \item $( \flatmill{\sigma} \seq \flatmill{\tau} ) \Rightarrow \flatmill{(\sigma \mto \tau)}$ \label{eq:atomic-rule-in-MILL-implication-intro}
        \item $( \seq \flatmill{(\sigma \mto \tau)} ), (\seq \flatmill{\sigma}) \Rightarrow \flatmill{\tau}$ \label{eq:atomic-rule-in-MILL-implication-elim}
        \item $(\seq \flatmill{\sigma}), (\seq \flatmill{\tau}) \Rightarrow \flatmill{(\sigma \mand \tau)}$ \label{eq:atomic-rule-in-MILL-conjunction-intro}
        \item $(\seq \flatmill{ (\sigma \mand \tau) }), (\flatmill{\sigma} \mcomma \flatmill{\tau} \seq \at{p}) \Rightarrow \at{p} $ 
        \label{eq:atomic-rule-in-MILL-conjunction-elim-atom}
        %
        %
        \item $\Rightarrow \flatmill{\mtop}$ \label{eq:atomic-rule-in-MILL-mtop-intro} 
        \item $( \seq \flatmill{\mtop}), (\seq \flatmill{\tau}) \Rightarrow \flatmill{\tau} $ \label{eq:atomic-rule-in-MILL-mtop-elim} 
    \end{enumerate}
    The following two statements are the key to completeness: 
    \begin{description}
        \item[$\dagger$\label{eq:mill-complete-flat-iff}] For every $\xi \in \Xi$, every $\baseX \baseGeq \baseMill$ and every $\at{U}$, $\suppM{\baseX}{\at{U}} \flatmill{\xi}$ iff $\suppM{\baseX}{\at{U}} \xi$. 
        \item[$\ddagger$\label{eq:mill-complete-deflat}] For every atomic multiset $\at{P}$ and atom $\at{q}$, if $\at{P} \deriveBaseM{\baseMill} \at{q}$ then $\deflatmill{\at{P}} \provesMILL \deflatmill{\at{q}}$. 
    \end{description}

    Starting from our assumption $\Gamma \suppM{}{} \phi$, we can conclude $\flatmill{\Gamma} \suppM{\baseMill}{} \flatmill{\phi}$ as follows: starting from arbitrary base $\baseB \baseGeq \baseMill$ and atomic multisets $\at{U}_1, \dots \at{U}_n$ satisfying $\suppM{\baseB}{\at{U}_1} \flatmill{\gamma_1}, \dots, \suppM{\baseB}{\at{U}_n} \flatmill{\gamma_n}$, by (the `only if' direction of) \eqref{eq:mill-complete-flat-iff} we have $\suppM{\baseB}{\at{U}_1} \gamma_1, \dots, \suppM{\baseB}{\at{U}_n} \gamma_n$; by the assumption $\Gamma \suppM{}{} \phi$, it follows that $\suppM{\baseB}{\at{U}} \phi$ where $\at{U} = \at{U}_1 \mcomma \dots \mcomma \at{U}_n$; applying \eqref{eq:mill-complete-flat-iff} again (but this time using the `if' direction) we know $\suppM{\baseB}{\at{U}} \flatmill{\phi}$. Then, according to Proposition~\ref{prop:mill-atomic-sound-and-complete}, $\flatmill{\Gamma} \suppM{\baseMill}{} \flatmill{\phi}$ implies $\flatmill{\Gamma} \deriveBaseM{\baseMill} \flatmill{\phi}$. So, by \eqref{eq:mill-complete-deflat}, $\deflatmill{(\flatmill{\Gamma})} \provesMILL \deflatmill{(\flatmill{\phi})}$, which according to the definition of $\flatmill{(\cdot)}$ and $\deflatmill{(\cdot)}$ says $\Gamma \provesMILL \phi$. 

    \medskip
    So it only remains to prove \eqref{eq:mill-complete-flat-iff} and \eqref{eq:mill-complete-deflat}. 

    We first look at \eqref{eq:mill-complete-flat-iff}. We fix an arbitrary base $\baseB \baseGeq \baseMill$ and atomic multiset $\at{S}$, and prove by induction on the structure of $\xi$. 
    \begin{itemize}[label={-}]
        \item $\xi$ is atomic. Then by definition, $\flatmill{\xi} = \xi$, so \eqref{eq:mill-complete-flat-iff} is a tautology. 
        \item $\xi$ is $\mtop$. 
        \begin{align*}
            \suppM{\baseB}{\at{S}} \flatmill{\mtop}
            & \text{ iff } \at{S} \deriveBaseM{\baseB} \flatmill{\mtop} & \eqref{eq:MILL-BeS-At} \\
            & \text{ iff for every } \baseX \baseGeq \baseB, \at{U}, \at{p}, \text{ if } \at{U} \deriveBaseM{\baseX} \at{p}, \text{ then } \at{S} \mcomma \at{U} \deriveBaseM{\baseX} \at{p} & (\text{Lemma~\ref{lem:flat-derive-correct}}) \\ 
            & \text{ iff for every } \baseX \baseGeq \baseB, \at{U}, \at{p}, \text{ if } \suppM{\baseX}{\at{U}} \at{p}, \text{ then } \suppM{\baseX}{\at{S} \mmcomma \at{U}} \at{p} & \eqref{eq:MILL-BeS-At} \\ 
            & \text{ iff } \suppM{\baseB}{\at{S}} \mtop & \eqref{eq:MILL-BeS-I}
        \end{align*}
        \item $\xi$ is of the form $\sigma \mto \tau$. By the construction of $\Xi$, $\sigma$ and $\tau$ are both in $\Xi$ as well, so IH applies. Therefore, 
        \begin{align*}
            \suppM{\baseB}{\at{S}} \flatmill{(\sigma \mto \tau)} & \text{ iff } \flatmill{\sigma} \suppM{\baseB}{\at{S}} \flatmill{\tau} \tag{Lemma~\ref{lem:flat-derive-correct} and Proposition~\ref{prop:mill-atomic-sound-and-complete}} \\
            & \text{ iff } \sigma \suppM{\baseB}{\at{S}} \tau \tag{IH} \\
            & \text{ iff } \suppM{\baseB}{\at{S}} \sigma \mto \tau  \tag{\ref{eq:MILL-BeS-imply}} 
        \end{align*}
        %
        \item $\xi$ is of the form $\sigma \mand \tau$. Again we can use IH on $\sigma$ and $\tau$ as both are in $\Xi$ as well. Therefore, 
        \begin{align*}
            \suppM{\baseB}{\at{S}} \flatmill{(\sigma \mand \tau)} & \text{ iff for every } \baseX \baseGeq \baseX, \at{U}, \at{p}, \text{ if } \flatmill{\sigma} \mcomma \flatmill{\tau} \suppM{\baseX}{\at{U}} \at{p}, \text{ then } \suppM{\baseX}{\at{S} \mmcomma \at{U}} \at{p} \tag{Lemma~\ref{lem:flat-derive-correct} and Proposition~\ref{prop:mill-atomic-sound-and-complete}} \\
            & \text{ iff for every } \baseX \baseGeq \baseX, \at{U}, \at{p}, \text{ if } \sigma \mcomma \tau \suppM{\baseX}{\at{U}}, \text{ then } \suppM{\baseX}{\at{S} \mmcomma \at{U}} \at{p} \tag{IH} \\
            & \text{ iff } \suppM{\baseB}{\at{S}} \sigma \mand \tau \tag{\ref{eq:MILL-BeS-tensor}}. 
        \end{align*}
    \end{itemize}
    This completes the proof by induction on $\xi$ for \eqref{eq:mill-complete-flat-iff}. 

    Next we turn to showing \eqref{eq:mill-complete-deflat}. 
    By the inductive definition of $\deriveBaseM{\baseMill}$ (see Definition~\ref{def:derivability-base-MILL}), it suffices to show the follows: 
    \begin{equation}
    \label{eq:mill-complete-deflat-proof-1}
        \deflatmill{\at{p}} \provesMILL \deflatmill{\at{p}} 
    \end{equation}
    \begin{equation}
    \label{eq:mill-complete-deflat-proof-2}
        \begin{split}
            \text{ If } & \left( (\at{P}_1 \seq \at{q}_1), \dots, (\at{P}_n \seq \at{q}_n) \Rightarrow \at{r} \right) \in \baseMill, \\
            & \text{ and } \deflatmill{\at{S}_1} \mcomma \deflatmill{\at{P}_1} \provesMILL \deflatmill{\at{q}_1}, \dots, \deflatmill{\at{S}_n} \mcomma \deflatmill{\at{P}_n} \provesMILL \deflatmill{\at{q}_n}, \\
            & \text{ then } \deflatmill{\at{S}_1} \mcomma \dots \mcomma \deflatmill{\at{S}_n} \provesMILL \deflatmill{\at{r}}. 
        \end{split} 
    \end{equation}
    Now \eqref{eq:mill-complete-deflat-proof-1} follows imemdiately from \eqref{eq:mill-axiom}. As for \eqref{eq:mill-complete-deflat-proof-2}, we simply need to prove the statement for each atomic rule in base $\baseMill$, which according to the definition of $\baseMill$ 
    amounts to proving the following facts:  
    \begin{itemize}[label=-]
        \item Suppose $(\flatmill{\sigma} \seq \flatmill{\tau}) \Rightarrow \flatmill{(\sigma \mto \tau)}$ is in $\baseMill$, and $\deflatmill{\at{S}}\mcomma \deflatmill{(\flatmill{\sigma})} \provesMILL \deflatmill{ (\flatmill{\tau}) }$, we show $\deflatmill{S} \provesMILL \deflatmill{ (\flatmill{(\sigma \mto \tau)}) }$. By the definition of $\deflatmill{(\cdot)}$, $\deflatmill{\at{S}} \mcomma \deflatmill{(\flatmill{\sigma})} \provesMILL \deflatmill{ (\flatmill{\tau}) }$ is $\deflatmill{\at{S}} \mcomma \sigma \provesMILL \tau$, and the goal $\deflatmill{S} \provesMILL \deflatmill{ (\flatmill{(\sigma \mto \tau)}) }$ is $\deflatmill{S} \provesMILL \sigma \mto \tau$., which follows immediately from \eqref{eq:mill-implication-intro}. 
        \item Suppose $( \seq \flatmill{(\sigma \mto \tau)} ), (\seq \flatmill{\sigma}) \Rightarrow \flatmill{\tau}$ is in $\baseMill$. We show that $\deflatmill{\at{S}_1} \provesMILL \deflatmill{(\flatmill{(\sigma \mto \tau)})}$ and $\deflatmill{\at{S}_2} \provesMILL \deflatmill{(\flatmill{\sigma})}$ implies $\deflatmill{\at{S}_1} \mcomma \deflatmill{\at{S}_2} \provesMILL \deflatmill{(\flatmill{\tau})}$. This is equivalent to that $\deflatmill{\at{S}_1} \provesMILL \sigma \mto \tau$ and $\deflatmill{\at{S}_2} \provesMILL \sigma$ implies $\deflatmill{\at{S}_1} \mcomma \deflatmill{\at{S}_2} \provesMILL \tau$, which follows immediately from \eqref{eq:mill-implication-elim}. 
        \item Suppose $(\seq \flatmill{\sigma}), (\seq \flatmill{\tau}) \Rightarrow \flatmill{(\sigma \mand \tau)}$ is in $\baseMill$. We show that $\deflatmill{\at{S}_1} \provesMILL \deflatmill{ (\flatmill{\sigma}) }$ and $\deflatmill{\at{S}_2} \provesMILL \deflatmill{ (\flatmill{\tau}) }$ implies $\deflatmill{\at{S}_1} \mcomma \deflatmill{\at{S}_2} \provesMILL \deflatmill{ (\flatmill{ (\sigma \mand \tau) }) }$. According to the definition of $\deflatmill{(\cdot)}$, this is equivalent to that $\deflatmill{\at{S}_1} \provesMILL \sigma$ and $\deflatmill{\at{S}_2} \provesMILL \tau$ implies $\deflatmill{\at{S}_1} \mcomma \deflatmill{\at{S}_2} \provesMILL \sigma \mand \tau$, which follows immediately from \eqref{eq:mill-conjunction-intro}. 
        \item Suppose $(\seq \flatmill{ (\sigma \mand \tau) }), (\flatmill{\sigma} \mcomma \flatmill{\tau} \seq \at{p}) \Rightarrow \at{p} $ is in $\baseMill$. We show that $\deflatmill{\at{S}} \provesMILL \deflatmill{ (\flatmill{ (\sigma \mand \tau) }) }$ together with $\deflatmill{\at{T}} \mcomma \deflatmill{ (\flatmill{\sigma}) } \mcomma \deflatmill{ (\flatmill{\tau}) } \provesMILL \deflatmill{\at{p}}$ implies $\deflatmill{\at{S}} \mcomma \deflatmill{\at{T}} \provesMILL \deflatmill{\at{p}}$. According to the definition of $\deflatmill{(\cdot)}$, this is equivalent to that $\deflatmill{\at{S}} \provesMILL \sigma \mand \tau$ and $\deflatmill{\at{T}} \mcomma \sigma \mcomma \tau \provesMILL \deflatmill{\at{p}}$ implies $\deflatmill{\at{S}} \mcomma \deflatmill{\at{T}} \provesMILL \deflatmill{\at{p}}$, which follows immediately from \eqref{eq:mill-conjunction-elim}. 
        %
        %
        \item Suppose $\Rightarrow \mtop$ is in $\baseMill$, then we show $\provesMILL \deflatmill{\mtop}$, namely $\provesMILL \mtop$. And this is exactly \eqref{eq:mill-top-intro}. 
        \item Suppose $(\seq \flatmill{ \mtop }), (\seq \flatmill{\tau}) \Rightarrow \flatmill{\tau}$ is in $\baseMill$, we show that $\deflatmill{\at{S}_1} \provesMILL \deflatmill{ (\flatmill{\mtop}) }$ together with $\deflatmill{\at{S}_2} \provesMILL \deflatmill{ (\flatmill{\tau}) }$ implies $\deflatmill{\at{S}_1} \mcomma \deflatmill{\at{S}_2} \provesMILL \deflatmill{ (\flatmill{\tau}) }$. By the definition of $\deflatmill{ (\cdot) }$, this is equivalent to that $\deflatmill{\at{S}_1} \provesMILL \mtop$ and $\deflatmill{\at{S}_2} \provesMILL \tau$ implies $\deflatmill{\at{S}_1} \mcomma \deflatmill{\at{S}_2} \provesMILL \tau$. This follows from \eqref{eq:mill-top-elim}. 
    \end{itemize}
    This completes the case analysis for establishing \eqref{eq:mill-complete-deflat}. 
\end{proof}
\begin{lemma}
\label{lem:flat-derive-correct}
    The following holds for arbitrary base $\baseB \baseGeq \baseMill$ and atomic multiset $\at{S}$, when $\sigma \mto \tau$, $\sigma \mand \tau$, or $\mtop$ is in $\Xi$, respectively: 
    \begin{enumerate}
        \item $\at{S} \deriveBaseM{\baseB} \flatmill{(\sigma \mto \tau)}$ iff $\at{S} \mcomma \flatmill{\sigma} \deriveBaseM{\baseB} \flatmill{\tau}$. \label{eq:flat-derive-correct-1}
        \item $\at{S} \deriveBaseM{\baseB} \flatmill{(\sigma \mand \tau)}$ iff for every $\baseY \baseGeq \baseB$, $\at{V}$, $\at{p}$, if $\at{V} \mcomma \flatmill{\sigma} \mcomma \flatmill{\tau} \deriveBaseM{\baseY} \at{p}$, then $\at{S} \mcomma \at{V} \deriveBaseM{\baseY} \at{p}$. \label{eq:flat-derive-correct-2} 
        \item $\at{S}  \deriveBaseM{\baseB} \flatmill{\mtop}$ iff for every $\baseY \baseGeq \baseB$, $\at{V}$, $\at{p}$, if $\at{V} \deriveBaseM{\baseY} \at{p}$, then $\at{S} \mcomma \at{V} \deriveBaseM{\baseY} \at{p}$. \label{eq:flat-derive-correct-3} 
    \end{enumerate}
\end{lemma}
\begin{proof}
    Let us fix arbitrary base $\baseB \baseGeq \baseMill$ and atomic multiset $\at{S}$. 
    \begin{enumerate}
        \item We prove the two directions separately. 
        \begin{itemize}
            \item Left to right: We assume $\at{S} \deriveBaseM{\baseB} \flatmill{(\sigma \mto \tau)}$. Note that $\flatmill{\sigma} \deriveBaseM{\baseB} \flatmill{\sigma}$ by \eqref{eq:derive-mill-ref}. Also, the atomic rule $\left( \seq \flatmill{(\sigma \mto \tau)}, \seq \flatmill{\sigma} \right) \Rightarrow \flatmill{\tau}$ is in $\baseMill$ thus in $\baseB$. Therefore, by \eqref{eq:derive-mill-app} we can conclude $\at{S} \mcomma \flatmill{\sigma} \deriveBaseM{\baseB} \flatmill{\tau}$. 
            \item Right to left: We assume $\at{S} \mcomma \flatmill{\sigma} \deriveBaseM{\baseB} \flatmill{\tau}$. Together with that $(\flatmill{\sigma} \seq \flatmill{\tau}) \Rightarrow \flatmill{(\sigma \mto \tau)}$ is in $\baseMill$ thus in $\baseB$, it follows from \eqref{eq:derive-mill-app} that $\at{S} \deriveBaseM{\baseB} \flatmill{(\sigma \mto \tau)}$. 
        \end{itemize}
        \item Again we show the two directions separately. 
        \begin{itemize}
            \item Left to right: We assume $\at{S} \deriveBaseM{\baseB} \flatmill{(\sigma \mand \tau)}$. It suffices to fix some $\baseC \baseGeq \baseB$, $\at{T}$ and $\at{q}$ satisfying $\at{T} \mcomma \flatmill{\sigma} \mcomma \flatmill{\tau} \deriveBaseM{\baseC} \at{q}$, and then show $\at{S} \mcomma \at{T} \deriveBaseM{\baseC} \at{q}$. Note that the atomic rule $\left( \seq \flatmill{(\sigma \mand \tau)}, \flatmill{\sigma} \mcomma \flatmill{\tau} \seq \at{q} \right) \Rightarrow \at{q}$ is in $\baseB$ thus also in $\base{C}$, therefore from the two assumptions we can derive $\at{S} \mcomma \at{T} \deriveBaseM{\baseC} \at{q}$. 
            \item Right to left: We assume that for every $\baseY \baseGeq \baseB$, $\at{V}$, and $\at{p}$, if $\at{V} \mcomma \flatmill{\sigma} \mcomma \flatmill{\tau} \deriveBaseM{\baseY} \at{p}$, then $\at{S} \mcomma \at{V} \deriveBaseM{\baseY} \at{p}$. The goal is to show $\at{S} \deriveBaseM{\baseB} \flatmill{(\sigma \mand \tau)}$. In particular, suppose we have $\flatmill{\sigma} \mcomma \flatmill{\tau} \deriveBaseM{\baseB} \flatmill{(\sigma \mand \tau)}$, then $\at{S} \deriveBaseM{\baseB} \flatmill{(\sigma \mand \tau)}$ immediately follows from the assumption. To show $\flatmill{\sigma} \mcomma \flatmill{\tau} \deriveBaseM{\baseB} \flatmill{(\sigma \mand \tau)}$, it suffices to apply \eqref{eq:derive-mill-app} to the atomic rule $\left( \seq \flatmill{\sigma}, \seq \flatmill{\tau} \right) \Rightarrow \flatmill{(\sigma \mand \tau)}$ in $\baseB$ as well as the fact that both $\flatmill{\sigma} \deriveBaseM{\baseB} \flatmill{\sigma}$ and $\flatmill{\tau} \deriveBaseM{\baseB} \flatmill{\tau}$ hold (using \eqref{eq:derive-mill-ref}). 
        \end{itemize}
        \item We prove the two directions separately. 
        \begin{itemize}
            \item Left to right: We fix some $\baseC \baseGeq \baseB$, $\at{T}$, and $\at{q}$ such that $\at{T} \deriveBaseM{\baseC} \at{q}$, and the goal is to show that $\at{S} \mcomma \at{T} \deriveBaseM{\baseC} \at{q}$ holds. Notice that the atomic rule $\left( \seq \flatmill{\mtop}, \seq \flatmill{\tau} \right) \Rightarrow \flatmill{\tau}$ is in $\baseB$ thus in $\baseC$, so apply \eqref{eq:derive-mill-app} to this rule together with $\at{S} \deriveBaseM{\baseC} \flatmill{\mtop}$ (immediate consequence of $\at{S} \deriveBaseM{\baseB} \flatmill{\mtop}$ and $\baseC \baseGeq \baseB$) and $\at{T} \deriveBaseM{\baseC} \at{q} $ entails that $\at{S} \mcomma \at{T} \deriveBaseM{\baseC} \at{q}$. 
            \item Right to left: This is the simpler direction. Since the atomic rule $\Rightarrow \flatmill{\mtop}$ is in $\baseB$, using \eqref{eq:derive-mill-app} we have $\deriveBaseM{\baseB} \flatmill{\mtop}$. The RHS of the statement entails that $\at{S} \deriveBaseM{\baseB} \flatmill{\mtop}$. 
        \end{itemize}
    \end{enumerate}
    This completes the proof for all the three statements. 
\end{proof}

\end{document}